\documentclass[a4paper,twocolumn,10pt]{article}

\usepackage{authblk}

\usepackage{amsfonts}
\usepackage{amsmath}
\usepackage{amssymb}
\usepackage{amsthm}
\usepackage{mathtools}
\usepackage{bm}
\usepackage{graphicx}
\usepackage{color}
\usepackage{subfig}
\usepackage{url}
\usepackage[ruled, linesnumbered, noend]{algorithm2e}
\usepackage{stfloats}
\usepackage{tikz}
\usetikzlibrary{patterns,decorations.pathmorphing}
\usepackage{caption}
\usepackage{multirow}
\usepackage{blindtext}
\usepackage{adjustbox}

\newtheorem{definition}{Definition}
\newtheorem{lemma}{Lemma}
\newtheorem{theorem}{Theorem}
\newtheorem{corollary}{Corollary}
\newtheorem{proposition}{Proposition}
\newtheorem{claim}{Claim}

\definecolor{myGrey}{RGB}{160,160,160}

\begin{document}

\title{Protect Edge Privacy in Path Publishing with Differential Privacy}
\author[1]{Zhigang Lu}
\author[1,2]{Hong Shen}
\affil[1]{School of Computer Science, The University of Adelaide, Adelaide, Australia\\\{zhigang.lu,hong.shen\}@adelaide.edu.au}
\affil[2]{School of Data and Computer Science, Sun Yat-sen University, Guangzhou, China}
\date{}

\maketitle

\begin{abstract}
Paths in a given network (map) are a generalised form of time-serial chains in many real world applications, such as trajectories and Internet flows. Differentially private trajectory publishing concerns publishing path information that is usable to the genuine users yet secure against adversaries to reconstruct the path with maximum background knowledge. The exiting studies all assume this knowledge to be all but one vertex on the path, i.e., the adversaries have the knowledge of the network topology and the path but one missing vertex (and its two connected edges) on the path. To prevent the adversaries recovering the missing information, they publish a perturbed path where each vertex is sampled from a pre-defined set with differential privacy (DP) to replace the corresponding vertex in the original path. In this paper, we relax this assumption to be all but one edge on the path, and hence consider the scenario of more powerful adversaries with the maximum background knowledge of the entire network topology and the path (including all the vertices) except one (arbitrary) missing edge. Under such assumption, the perturbed path produced by the existing work is vulnerable, because the adversary can reconstruct the missing edge from the existence of an edge in the perturbed path whose two ends are close to the two ends of the missing edge. To address this vulnerability and effectively protect edge privacy, instead of publishing a perturbed path, we propose a novel scheme of graph-based path publishing to protect the original path by embedding the path in a graph that contains fake edges and replicated vertices applying the differential privacy technique, such that only the legitimate users who have the full knowledge of the network topology are able to recover the exact vertices and edges of the original path with high probability. We theoretically analyse the performance of our algorithm in terms of output quality — differential privacy and utility, and execution efficiency. We also conduct extensive experimental evaluations on a high-performance cluster system to validate our analytical results.
\end{abstract}

\section{Introduction}
\label{sec:6:introduction}
Paths are pervasive in our daily life, such as, geometric trajectory, internet traffic flow, supply chain, and intelligence exchange network~\cite{FouldsL2012,GuoL2016}. Publishing a path helps the owners of the path discover more knowledge about the path and obtain the information about their roles on the path. However, attackers may infer the privacy of a path from its published information by applying their auxiliary information~\cite{ChenS2016,MuX2019}. Therefore, it is crucial to consider privacy disclosure risks when publishing a path.

In this paper, we consider path publishing against adversaries with the maximum knowledge of the network topology but an arbitrary edge to infer the path (edge).
That is, given a network $N = (V, E)$, we publish a  path $P = (V, E_{P} \subset E)$ protecting an arbitrary edge $e \in E_{P}$ against adversaries who have $N^{\prime} = (V, E\setminus\{e\})$ containing $P^{\prime} = (V, E_{P}\setminus\{e\})$ containing $P^{\prime} = (V, E_{P}\setminus\{e\})$. This problem can be found in several real-life scenarios:

\textbf{Supply Chain}. In this scenario, $N$ is a service system of an organisation; $V$ is the set of service suppliers; $E$ is the set of resource transfer links; $P$ is a supply chain; path publishing is the way to ensure every supplier knows how resources are transferred; "adversaries" are competing organisations which want to know how the target organisation allocates the resources.

\textbf{Intelligence Exchange}. In this scenario, $N$ is a spy network; $V$ is the set of spies; $E$ is the set of spy communication channels; $P$ is the message exchange path; path publishing is the way to let all spies know whom should be contacted with; "adversaries" are counter-intelligence units who want to know how a targeted spy exchanges information with others.

Presently, privacy-preserving path publishing against the maximum background knowledge is well studied in trajectory publishing with differential privacy (DP)~\cite{DWORK2006B} because DP is a promising privacy notion due to its privacy-preserving guarantee against the attackers with arbitrary auxiliary information. Particularly, there are two tracks of research in differentially private trajectory publishing: set of trajectories publishing~\cite{ChenR2012,HuaJ2015,NieY2016,LiC2017,CaoY2018,AlhussaeniK2018,GursoyM2018,MaX2018} and (single) trajectory publishing~\cite{JiangK2013,XiaoY2015,HeX2015,WangS2016,WangS2017,CaoY2019}. In this paper we address the problem in the latter track.

In a nutshell, existing studies in the differentially private trajectory publishing share a similar idea in two-fold. First of all, they take the same assumption on the adversaries' background knowledge. That is, the adversaries have full information about the map (network) except some locations (and all edges incident to them) in the target trajectory. In the worst case, only one location is missing at the attacker's side. Second, they use exponential mechanism of DP (ExpDP) to implement differential privacy. Specifically, the published trajectory contains fake locations only, which are sampled via ExpDP from a predefined area, to replace the real ones in the target trajectory. The differences over these methods mainly come from the quality measurement for sampling the fake locations due to different objectives, such as the temporal-spatio relationships between locations~\cite{XiaoY2015,CaoY2019}, the utility of the fake locations, including movement directions in a trajectory~\cite{JiangK2013} and similarity between the real and the fake locations~\cite{HeX2015,WangS2016}. 

The state-of-the-art results in single trajectory publishing achieve a satisfying trade-off between utility and privacy by DP against adversaries with missing vertices; however, we discover two weaknesses in the existing work~\cite{JiangK2013,HeX2015,XiaoY2015,WangS2016,CaoY2019}. First, the assumption of missing vertices does not cover the worst scenario of adversary' maximum background knowledge, i.e., an adversary has the full knowledge of the map including all vertices and edges except one edge on the trajectory, which arises in many real-life applications requiring protection of an sensitive connection between two entities such as supply chain, financial trading and intelligence exchange.   In this case, the published (perturbed, privacy-preserved) path by the existing work is vulnerable to preserve the existence of an arbitrary edge of the original path, because the adversary can reconstruct the missing edge by the existence of an edge in the perturbed path whose two ends are close to the two ends of the missing edge. Figure~\ref{fig:6:attack} depicts a running attack against existing work. Second, the existing work does not in general publish real vertices on the path, so it is hard for the genuine users (who have full information about the network) to reconstruct the original path from the published one to confirm the exact links between the real vertices.
\begin{figure}[!th]
\centering
\begin{tikzpicture}[xscale=0.5, yscale=0.5]
\draw [pattern=north east lines] (-4,1) circle [radius=0.22];
\node [above] at (-4,1+0.1) {$a^{\prime}$};
\draw [->] (-4+0.2,1-0.1) to (-2-0.2,0+0.1);
\draw [pattern=north east lines] (-2,0) circle [radius=0.22];
\node [below] at (-2,-0.1) {$b^{\prime}$};
\draw [->] (-2+0.22,0) to (0-0.22,0);
\draw [pattern=north east lines] (0,0) circle [radius=0.22];
\node [above] at (0,0.1) {$c^{\prime}$};
\draw [->] (0+0.16,0+0.16) to (2-0.16,2-0.16);
\draw [pattern=north east lines] (2,2) circle [radius=0.22];
\node [above] at (2,2+0.1) {$d^{\prime}$};

\draw (-5,0) to (-2,1);
\draw (1,-1) to (1,2);
\draw [fill=white] (-5,0) circle [radius=0.22];
\node [below] at (-5,-0.1) {$a$};
\draw [fill=white] (-2,1) circle [radius=0.22];
\node [above] at (-2,1+0.1) {$b$};
\draw [fill=white] (1,-1) circle [radius=0.22];
\node [below] at (1,-1-0.1) {$c$};
\draw [fill=white] (1,2) circle [radius=0.22];
\node [above] at (1,2+0.1) {$d$};
\draw [rotate around={135:(-4.5,0.5)},dashed] (-4.5,0.5) ellipse (28pt and 54pt);
\draw [dashed] (-2,0.5) ellipse (23pt and 50pt);
\draw [rotate around={45:(0.5,-0.5)},dashed] (0.5,-0.5) ellipse (28pt and 54pt);
\draw [rotate around={90:(1.5,2.2)},dashed] (1.5,2) ellipse (28pt and 45pt);
\draw [pattern=north east lines] (3,0+0.2) circle [radius=0.22];
\draw [pattern=north east lines] (4,0+0.2) circle [radius=0.22];
\draw [->] (3+0.22,0+0.2) to (4-0.22,0+0.2);
\node [right] at (4.1,0+0.2) {: published path};
\draw (3,-0.7) to (4,-0.7);
\draw [fill=white] (3,-1+0.3) circle [radius=0.22];
\draw [fill=white] (4,-1+0.3) circle [radius=0.22];
\node [right] at (4.1,-1+0.3) {: attackers' partial path};
\draw [dashed] (3.5,-2+0.4) ellipse (20pt and 10pt);
\node [right] at (4.1,-2+0.4) {: guessing sampling set};
\node [right] at (4.7,-3+0.67) {by vertices closeness};
\draw [thick] (-1.2+0.2,-2.1) to (-1.2+0.2,-2.5);
\draw [thick] (-0.8+0.2,-2.1) to (-0.8+0.2,-2.5);
\draw [thick] (-1.55+0.2,-2.35) to (-1+0.2,-2.7);
\draw [thick] (-0.45+0.2,-2.35) to (-1+0.2,-2.7);
\node [below] at (0,-3.1) {\textbf{Missing Edge Discovered:}};
\node [below] at (0,-3.8) {Edge $(b,c)$ must be in the path};
\end{tikzpicture}
\caption{A Running Attack against the Existing Algorithms of Differentially Private Trajectory Publishing}
\label{fig:6:attack}
\end{figure}

To fill the research gap, this paper studies how to protect an arbitrary edge in path publishing. In summary, our main contributions are:
\begin{itemize}
  \item We address the problem of differentially private path publishing in a generalised form against inference attacks from adversaries with more background knowledge than the existing work. Different from the existing work where the adversary knows the network topology and $n - 1$ out of $n$ vertices of the path, we preserve privacy against adversaries who have the full knowledge of the network and all vertices on the path except one missing edge of the path (network);
  \item Our proposed method guarantees a high data utility and security (differential privacy). That is, the genuine users (who have the full knowledge of the network) are able to reconstruct the path but the adversaries are unable to infer the missing edge and hence reconstruct the path. In particular, we publish a graph-based path where fake edges and duplicate real vertices are created to embed the real path into a graph for publishing, such that only the genuine users are able to reconstruct the exact vertices and edges of the original path, and determine the visiting order of the vertices with high probability;
  \item We evaluate the key properties of the proposed method both theoretically and experimentally. Specifically, we mathematically analyse differential privacy guarantee and data utility of our scheme. We conduct experiments on a cluster system and extensively evaluate our scheme with synthetic datasets for various network sizes and densities.
\end{itemize}
To the best of our knowledge, this is the first work that considers edge privacy, i.e., the existence of an edge between two vertices as privacy, in differentially private path publishing. 

\textbf{Organisation.} The rest of this paper is organised as follows: In Section~\ref{sec:6:rw}, we discuss existing work on differentially private single trajectory publishing for both the advantages and disadvantages. Next, in Section~\ref{sec:6:preliminaries}, we give a brief introduction of the preliminaries of this paper, including differential privacy, path model and the attackers' background knowledge. Then we propose our differentially private single path publishing in Section~\ref{sec:6:main}. Mathematical analysis are also showed in this section. Afterwards, in Section~\ref{sec:6:exp}, we provide the experimental evaluation to the performance of our algorithms on a synthetic dataset. Finally, we conclude this paper in Section~\ref{sec:6:conclusion}.

\section{Related Work}
\label{sec:6:rw}
In this section, we briefly summarise the state-of-the-art work of differentially private single trajectory publishing in both advantages and disadvantages.

As far as we know, \cite{JiangK2013} is the first differentially private algorithm for single trajectory publishing. In \cite{JiangK2013}, Jiang et. al. proposed a Sampling Distance and Direction (SDD) mechanism, which publishes a partially perturbed trajectory. In the published trajectory of SSD, the first and the last locations come from the original trajectory directly. While, for each of the rest locations, say fake location $i$, were determined by a noisy distance and a noisy direction based on the real location $i-1$. Exponential mechanism of differential privacy (ExpDP) was used for sampling the noisy distance and direction. 

Xiao et. al.~\cite{XiaoY2015,XiaoY2017} observed that the temporal-spatio relationship may release privacy even though a perturbed trajectory was published. So a $\delta$-Location set for each location in the trajectory was prepared in \cite{XiaoY2015}. All the locations in the $\delta$-Location set satisfy the requirement for the temporal-spacial relationship. Then, the most possible location to replace the real one in the original trajectory will be sampled via ExpDP from the $\delta$-Location set. In addition, a Planar Isotropic Mechanism (PIM) was proposed to achieve optimal lower bound of DP for the $\delta$-Location set. In addition, Cao et. al.~\cite{CaoY2019} proposed a Private Spatio-Temporal Event (PriSTE) framework to further study the spatio-temporal privacy of \cite{XiaoY2015} for choosing better fake locations in a given trajectory.

Similar to \cite{JiangK2013}, He et. al.~\cite{HeX2015,HeX2016} proposed a Differentially Private Trajectories (DPT) system to select the most similar and possible location in a given area to replace a real location in the trajectory. The DPT system was achieved by three key components: hierarchical referencing, model selection, and direction weighted sampling. Specifically, hierarchical referencing firstly generates a bunch of possible trajectory candidates. Afterwards, model selection selects a trajectory with ExpDP. Finally, direction weighted sampling adjusts the directions of each location in the fake trajectory to meet the directions in the original trajectory. 

Wang et. al.~\cite{WangS2016} proposed a Private Trajectory Calibration System (PTCS), which applies a clustering-based method to find the fake locations to replace the real ones for the published trajectory. In PTCS, as a initialisation step, frequently visited locations of the user, whose trajectory will be published then, are clustered firstly. Then an anchor location in each cluster is picked as the representative of the cluster. For a give trajectory of the user, each location in this trajectory is replaced by its nearest anchor location to form the published trajectory.

In summary, by considering the features of a trajectory, such as the movement directions, temporal-spacial relationships, the state-of-the-art achieve satisfying trade-off between privacy guarantee and data utility for future study on the published trajectory. However, because of the published fake vertices, even the trusted path participates cannot recover the original path. Furthermore, due to their assumption on the attackers' background knowledge, these works are vulnerable under our assumption of the attackers, i.e., the attackers know the vertices in a path but miss some edges. The attackers with such background knowledge can definitely infer the privacy from the published path by existing work in differentially private single trajectory publishing.

\section{Preliminaries}
\label{sec:6:preliminaries}
In this section, we briefly introduce the notion of privacy used in this paper, i.e., differential privacy, the models for map and path, and our assumption on the adversaries' auxiliary information.

\subsection{Differential Privacy.} Differential Privacy (DP) is a famous notion of privacy in recent privacy-preserving research field~\cite{GUPTA2010,LuZ2017,WANG2018,LuZ2019}, which was firstly introduced and defined by Dwork et at.~\cite{DWORK2006A}. Informally, DP is a scheme that minimises the sensitivity of output for a given statistical operation on two neighbouring (differentiated in one arbitrary record to protect) datasets. That is, DP guarantees the presence or absence of any item in a dataset will be concealed to the adversaries with maximum auxiliary information.

In DP, the basic setting is a pair of neighbouring datasets $X$ and $X^{\prime}$, where $X^{\prime}$ contains the information of all the items except one item in a dataset $X$. A formal definition of Differential Privacy is shown as follow:
\begin{definition}[$\epsilon$-DP~\cite{DWORK2006A}]
\label{def:dp}
A randomised mechanism $\mathcal{A}$ is $\epsilon$-differentially private if for all neighbouring datasets $X$ and $X^{\prime}$ in domain $\mathcal{D}$, and for an arbitrary answer $s \in Range(\mathcal{A})$, $\mathcal{A}$ satisfies: 
$$\Pr[\mathcal{A}(X) = s]\leq \exp(\epsilon)\cdot\Pr[\mathcal{A}(X^{\prime}) = s]$$
where $\epsilon$ is the privacy budget.
\end{definition}

Two parameters are essential to DP: the privacy budget $\epsilon$ and the function sensitivity (global) $\Delta f$, i.e. $\Delta f(X)$, where $f$ is the query function to the dataset $X$. The privacy budget $\epsilon$ is set by the trusted dataset curator (who has full access to the dataset $X$). Theoretically, a smaller $\epsilon$ denotes a higher privacy guarantee because the privacy budget $\epsilon$ reflects the magnitude of the difference between two neighbouring datasets. $\Delta f$ is calculated by the following equation,
\begin{equation*}
\label{eq:6:ls}
\Delta f(X) = \max_{\forall X \in \mathcal{D}}|f(X) - f(X^{\prime})|,
\end{equation*}

In this paper, we mainly use the Exponential mechanism of DP (ExpDP)~\cite{MCSHERRY2007}. In general, the ExpDP is suitable for preserving privacy in the non-numeric computation. The ExpDP introduces a quality function $q(X, x)$ which reflects how appealing the pair $(X, x)$ is, where $X$ denotes a dataset and $x$ is the random respond to a query function on the dataset $X$. When applying the ExpDP, we can simply treat it as a weighted sampling, where the quality function assigns weights to the sample space. The formal definition of ExpDP is shown below:
\begin{definition}[Exponential Mechanism~\cite{MCSHERRY2007}]
\label{def:expMech}
Given a quality function of a dataset $X$, $q(X,x)$, which reflects the quality of a query respond $x$. The exponential mechanism $\mathcal{A}$ provides $\epsilon$-differential privacy, if $\mathcal{A}(X) = \left\{ \Pr[x] \propto \exp\left({\frac{\epsilon\cdot q(X,x)}{2\Delta q}}\right) \right\}$, where $\Delta{q}$ is the sensitivity of the quality function $q(X, x)$, $\epsilon$ is the privacy budget.
\end{definition}

\subsection{Network and Path Models}
Without loss of generality, a network (e.g., roadmap, internet node map, spy network, etc.) is a bidirected/undirected graph $N = (V, E)$, where $V$ is the set of vertices, $E$ is the set of edges. In this paper, the network is a connected graph to contain the path $P$.

The path $P = (V, E_{P} \subset E)$ is a subset of the network $N$. We define the path as a sequence of vertices in visiting order, i.e., a path $P = \{v_{1}, v_{2}, \dots, v_{n}\}$, where $v_{i} \in V$ is the $i$th vertex in $P$, $e_{i} = (v_{i}, v_{i+1}) \in E_{P}$ is the $i$th edge in $P$. Since the path $P$ may contain rings, we may have $\exists i \neq j$, $v_{i} (\in V) = v_{j} (\in V)$.

\begin{figure}[!th]
\centering
\begin{tikzpicture}[xscale=0.5, yscale=0.5]
\draw (0,0) to (3,0);
\draw [decorate,decoration={snake,post length=0.1mm}] (0,0) to (3,0);
\draw [dashed] (0,0) to [out=330,in=210] (3,0);
\draw (0,0) to [out=50,in=130] (9,0);
\draw (3,0) to (3,-3);
\draw [decorate,decoration={snake,post length=0.1mm}] (3,0) to (3,-3);
\draw [dashed] (3,0) to [out=210,in=150] (3,-3);
\draw (3,0) to (6,-3);
\draw (3,0) to (6,0);
\draw (3,0) to [out=50,in=130] (9,0);
\draw [myGrey] (3,-3) to (6,-3);
\draw [myGrey,decorate,decoration={snake,post length=0.1mm}] (3,-3) to (6,-3);
\draw (3,-3) to (6,0);
\draw (3,-3) to [out=290,in=300] (9,0);
\draw (6,-3) to (6,0);
\draw [decorate,decoration={snake,post length=0.1mm}] (6,-3) to (6,0);
\draw [dashed] (6,-3) to [out=30,in=330] (6,0);
\draw (6,-3) to (9,0);
\draw [dashed] (6,-3) to [out=250,in=240] (0,0);
\draw (6,0) to (9,0);
\draw [decorate,decoration={snake,post length=0.1mm}] (6,0) to (9,0);
\draw [dashed] (6,0) to [out=330,in=210] (9,0);
\draw [fill=white] (0,0) circle [radius=0.3];
\node [above] at (0,0.2) {$a$};
\draw [fill=white] (3,0) circle [radius=0.3];
\node [above] at (3,0.2) {$b$};
\draw [fill=white] (3,-3) circle [radius=0.3];
\node [left] at (2.9,-3) {$c$};
\draw [fill=white] (6,-3) circle [radius=0.3];
\node [right] at (6.1,-3) {$d$};
\draw [fill=white] (6,0) circle [radius=0.3];
\node [above] at (6,0.2) {$e$};
\draw [fill=white] (9,0) circle [radius=0.3];
\node [above] at (9,0.2) {$f$};
\node [right] at (9.5,1) {$E:$};
\draw (11,1) to (12.5,1);
\node [right] at (9.5,0) {$P: a \rightarrow f$,};
\draw [decorate,decoration={snake,post length=0.1mm}] (13,0) to (15,0);
\node [right] at (9.5,-1) {$E^{\prime}:E \setminus \{(c,d)\}$};
\node [right] at (9.5,-2) {$P^{\prime}:$ Maybe:};
\draw [dashed] (13.6,-2) to (15.3,-2);
\node [right] at (10.5,-3) {$f \rightarrow d \rightarrow a \rightarrow c$};
\end{tikzpicture}
\caption{An Example of Network, Path, and Adversaries' Background Knowledge}
\label{fig:6:model}
\end{figure}
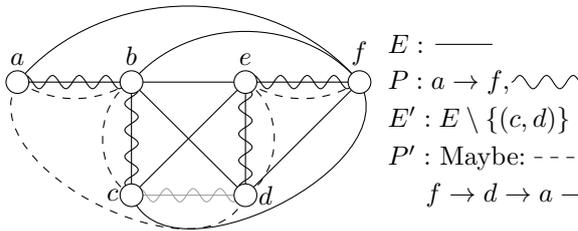

\subsection{Adversaries' Auxiliary Information}
\label{sec:6:background}
Different from the assumption of adversaries' auxiliary information in existing work of differentially private trajectory publishing, we assume the adversaries have correct information about the vertices but incorrect information about the edges of the path. We take the worst case of the adversaries' auxiliary information in this paper, so the adversaries do not know the existence of one edge in the original path $P$ (and the network $N$), say $e_{i}=(v_{i},v_{i+1})$. That is, the adversaries have $N^{\prime} = (V, E^{\prime})$, $P^{\prime} = (V, E_{P^{\prime}})$ that $e_{i} = E \setminus E^{\prime} = E_{P} \setminus E_{P^{\prime}}$. Figure~\ref{fig:6:model} illustrates the models of network, path and the adversaries' background knowledge.

\section{Proposed Algorithms}
\label{sec:6:main}
In this section, we will show our differentially private algorithm to preserve privacy in path publishing against the adversaries with background knowledge described in Section~\ref{sec:6:background}. In a nutshell, we propose a graph-based path publishing which uses the topology of the original network and the connections between vertices of the path to conceal the edge existence against the adversaries who do not have full knowledge about the edges in the given network.

\subsection{Idea Overview}
\label{sec:6:overview}
The idea of our graph-based differentially private path publishing was inspired by the permutation idea in \cite{GUPTA2010}. Especially, same as \cite{GUPTA2010}, the receivers of the published graph-based path are path participates, which can either be each vertex or each edge in the network. For example, in the intelligence exchange network, we aim to let all spies know whom should connect with, so the path participates are spies, i.e., the vertices in a network. In Internet traffic flow analysis, we want to monitor a given flow, so the path participates are the routers between servers, i.e., the edges in a network.
\begin{figure*}[!th]
\centering
\begin{tikzpicture}[xscale=0.5, yscale=0.5]
\node at (0,2) {$G$:};
\draw (0,0) to (3,-3);
\draw (0,0) to [out=50,in=130] (9,0);
\draw (3,0) to (6,-3);
\draw (3,0) to (6,0);
\draw (3,0) to [out=50,in=130] (9,0);
\draw (3,-3) to (6,0);
\draw (3,-3) to [out=315,in=270] (9,0);
\draw (6,-3) to (9,0);
\draw [fill=white] (0,0) circle [radius=0.3];
\node [above] at (0,0.2) {$a$};
\draw [fill=white] (3,0) circle [radius=0.3];
\node [above] at (3,0.2) {$b$};
\draw [fill=white] (3,-3) circle [radius=0.3];
\node [left] at (2.7,-3) {$c$};
\draw [fill=white] (6,-3) circle [radius=0.3];
\node [above] at (6,-2.8) {$d$};
\draw [fill=white] (6,0) circle [radius=0.3];
\node [above] at (6,0.2) {$e$};
\draw [fill=white] (9,0) circle [radius=0.3];
\node [above right] at (9,0) {$f$};

\node at (12,2) {$H$:};
\draw (12,0) to (15,0);
\draw [decorate,decoration={snake,post length=0.1mm}] (12,0) to (15,0);
\draw (12,0) to [out=270,in=235] (18,-3);
\draw (12,0) to [out=50,in=130] (18,0);
\draw (15,0) to (15,-3);
\draw [decorate,decoration={snake,post length=0.1mm}] (15,0) to (15,-3);
\draw (15,-3) to (18,-3);
\draw [decorate,decoration={snake,post length=0.1mm}] (15,-3) to (18,-3);
\draw (18,-3) to (18,0);
\draw [decorate,decoration={snake,post length=0.1mm}] (18,-3) to (18,0);
\draw (18,0) to (21,0);
\draw [decorate,decoration={snake,post length=0.1mm}] (18,0) to (21,0);
\draw [fill=white] (12,0) circle [radius=0.3];
\node [above] at (12,0.2) {$a$};
\draw [fill=white] (15,0) circle [radius=0.3];
\node [above] at (15,0.2) {$b$};
\draw [fill=white] (15,-3) circle [radius=0.3];
\node [above left] at (14.8,-3) {$c$};
\draw [fill=white] (18,-3) circle [radius=0.3];
\node [right] at (18.3,-3) {$d$};
\draw [fill=white] (18,0) circle [radius=0.3];
\node [above] at (18,0.2) {$e$};
\draw [fill=white] (21,0) circle [radius=0.3];
\node [above right] at (21,0) {$f$};

\node at (24,2) {$G_{P}$:};
\draw (24,-1) to (27,-3);
\draw (30,1) to (30,-1);
\draw (30,1) to (27,-1);
\draw (27,1) to (24,-1);
\draw (27,1) to (27,-1);
\draw (30,-1) to (27,-3);
\draw [fill=white] (24,-1) circle [radius=0.3];
\node [above] at (24,-0.8) {$a$};
\draw [fill=white] (30,1) circle [radius=0.3];
\node [above] at (30,1.2) {$b$};
\draw [fill=white] (27,1) circle [radius=0.3];
\node [above] at (27,1.2) {$c$};
\draw [fill=white] (30,-1) circle [radius=0.3];
\node [right] at (30.3,-1) {$d$};
\draw [fill=white] (27,-1) circle [radius=0.3];
\node [right] at (27.3,-1) {$e$};
\draw [fill=white] (27,-3) circle [radius=0.3];
\node [left] at (26.7,-3) {$f$};
\node at (15,-5) {Network $N$ contains an acyclic path $P$ (Figure~\ref{fig:6:model}) $\rightarrow$ $(G,H)$ $\rightarrow$ graph-based path $G_{P}$};
\end{tikzpicture}
\caption{An Example of Producing the Graph-based Path}
\label{fig:6:example}
\end{figure*}

In general, when building our graph-based path, we apply differential privacy to introduce uncertainties to hide the exact information about edges in the path with the following privacy-preserving rules for the published graph-based path.
\begin{enumerate}
  \item \label{itm:6:notsamebranch} For the edge $e = (u,v) \in E_{P}$, vertex $u$ and $v$ are not placed in a same branch in the graph-based path;
  \item \label{itm:6:samebranch} For the edge $e = (u,v) \in E \setminus E_{P}$, vertex $u$ and $v$ are placed in a same branch in the graph-based path;
  \item \label{itm:6:fakeedge} For the non-existent vertex connections $e = (u,v) \notin E$, vertex $u$ and $v$ are randomly (by DP) either in or not in a same branch in the graph-based path;
  \item \label{itm:6:fakevertex} For the vertices appear multiple times in the graph-based path, it does not implicitly indicate that those vertices are in a circle or not (because of DP);
  \item \label{itm:6:root} For a group of nodes without parent and a group of nodes without child in the graph-based path, we place the group, which contains the vertex closest to the first node of the path, as the root of the graph-based path.
\end{enumerate}

Particularly, because the path participates have full knowledge about the existence of edges, Rule~\ref{itm:6:notsamebranch} and Rule~\ref{itm:6:samebranch} guarantees the path participates can recover the exact vertices and edges of the path; Rule~\ref{itm:6:root} let the path participates determine the visiting order of the original path with high probability. In contrast, Rule~\ref{itm:6:fakeedge} and Rule~\ref{itm:6:fakevertex} introduce the uncertainties against the adversaries who do not have full knowledge about $P$, i.e., edge(s) missing in this paper.

\textbf{Path reconstruction}. The following steps help to recover the exact path from the published graph-based path. Firstly, each path participate confirms its status in the path $P$ according to Rule~\ref{itm:6:notsamebranch} and Rule~\ref{itm:6:samebranch}. Specifically, the vertices (as the path participates) check their neighbour vertices; the edges (as the path participates) check their two ends. Then the path participates exchange their status (possibly by secure multi-party computing~\cite{YAO1982}) to recognise the two ends of the path $P$. Finally, the path participates take one of the two ends, which is the closest node to the vertices as the roots of the graph-based path, as the starting point of the path. According to Rule~\ref{itm:6:root}, with high probability, each participate learns its index in the path, i.e., the visiting order of the path. 

\textbf{Producing the graph-based path}. To implement the privacy-preserving rules, in general, we split the give network $N$ to two sub-graphs: one (graph $G$) contains the edges not in $P$; one (graph $H$) contains the edges in $P$, but both of the two sub-graphs may contain some edges not in $N$. Then producing the graph-based path becomes the following problem. Given a graph $G = (V_{G}, E_{G})$, $U$ consists of all 2-element subsets of $V_{G}$, and $G$'s complement $H = (V_{H} = V_{G}, E_{H} = U \setminus E_{G})$, which is a Hamilton graph, e.g., $P$ is one Hamilton path in $H$. In this paper, we take $G$ as the input for producing $G_{P} = (V_{G}, E_{G_{P}})$, such that if $u$, $v$ are in a branch of $G_{P}$, then $e=(u, v) \in E_{G}$; otherwise, $e=(u, v) \in E_{H}$.

Note that, for the given network $N = (V, E)$, we have $(E \setminus E_{P}) \subset E_{G}$ because it is possible that $\exists e \in E_{G}$, that $e \notin E$ according to Rule~\ref{itm:6:fakeedge}; $V \subset V_{G}$ because it is possible $\exists v \in V_{G}$, that $v \notin V$ according to Rule~\ref{itm:6:fakevertex}. Figure~\ref{fig:6:example} depicts an example of producing $G_{P}$ from $G$ and $H$. Additionally, we should have an acyclic path in $H$; otherwise even the genuine path participates cannot recover the exact path. For example, consider using the original network in Figure~\ref{fig:6:prepmap} to produce a $G_{P}$, it is easy to conclude that  such a $G_{P}$ cannot recognise whether $e$ or $c$ is the succeeder of $b$ when we visit $b$ the first time.

Furthermore, $G$ and $G_{P}$ are equivalent for a graph-based path publishing. $G$ hides the path information, and potentially indicates the visiting order with Rule~\ref{itm:6:fakevertex}; however it causes high computational cost for recovering the real path. Corollary~\ref{cor:6:recover} proves that it takes $\mathcal{O}(|V| + |E|)$ for each path participate to recover its status in the path from $G$; in contrast, $\mathcal{O}(|V|)$ from $G_{P}$. Especially, in the worst case, $\mathcal{O}(|V| + |E|)$ is equivalent to $\mathcal{O}(|V|^{2})$, i.e., $\mathcal{O}(|V|) \ll \mathcal{O}(|V| + |E|)$. 

In addition, Table~\ref{tab:6:terms} lists the notations used in this paper.
\begin{table*}[!th]
\caption{The summary of notations.}
\label{tab:6:terms}
\centering
\begin{tabular}{c|l}
\hline
Notation & Description\\
\hline
\hline
$\cdot^{\prime}$ & The corresponding notation of adversaries' background knowledge\\
\hline
$\hat{\cdot}$ & The processed $\cdot$\\
\hline
$E$; $E_{G}$; $E_{P}$ & The set of edges in $N$; in $G$; in $P$\\
\hline
$G$ & The input graph to produce $G_{P}$\\
\hline
$G_{P}$ & The published graph-based path of this paper\\
\hline
$H$ & The complement graph of G\\
\hline
$N$ & The network contains the path $P$\\
\hline
$P$ & The original path for publishing\\
\hline
$V$; $V_{G}$ & The set of vertices in $N$ and $P$; in $G$\\
\hline
\end{tabular}
\end{table*}

\subsection{Map Pre-processing}
\label{sec:6:prep}
In this section, we pre-process the given network to meet our privacy-preserving Rule~\ref{itm:6:fakeedge}, Rule~\ref{itm:6:fakevertex} and the additionally acyclic requirement for graph $H$ by injecting duplicate vertices and fake edges with the exponential mechanism of DP (ExpDP).

\subsubsection{Differentially Private Vertices Pre-processing}
To convert a network to fit Rule~\ref{itm:6:fakevertex} and the acyclic requirement, in this section, we borrow the inheritance concept in objective-oriented programming (OOP) to create base vertices and sub-vertex(s) in a network. Such an idea comes from a feature of a cyclic path, that is, when there is a cyclic path in the network, at least one vertex is repeatedly visited.

\begin{figure*}[!ht]
\centering
\begin{tikzpicture}[xscale=0.5, yscale=0.5]
\node at (0,2) {The Give Network:};
\draw (0,0) to (3,0);
\draw [decorate,decoration={snake,post length=0.1mm}] (0,0) to (3,0);
\draw (0,0) to (3,-3);
\draw (3,0) to (3,-3);
\draw [decorate,decoration={snake,post length=0.1mm}] (3,0) to (3,-3);
\draw (3,0) to (6,0);
\draw [decorate,decoration={snake,post length=0.1mm}] (3,0) to (6,0);
\draw (3,0) to [out=45,in=135] (9,0);
\draw (3,-3) to [out=315,in=270] (9,0);
\draw (3,-3) to (6,-3);
\draw [decorate,decoration={snake,post length=0.1mm}] (3,-3) to (6,-3);
\draw (6,0) to (6,-3);
\draw [decorate,decoration={snake,post length=0.1mm}] (6,0) to (6,-3);
\draw (6,0) to (9,0);
\draw [decorate,decoration={snake,post length=0.1mm}] (6,0) to (9,0);
\draw [fill=white] (0,0) circle [radius=0.3];
\node [above] at (0,0.2) {$a$};
\draw [fill=white] (3,0) circle [radius=0.3];
\node [above] at (3,0.2) {$b$};
\draw [fill=white] (3,-3) circle [radius=0.3];
\node [left] at (2.8,-3) {$e$};
\draw [fill=white] (6,-3) circle [radius=0.3];
\node [right] at (6.2,-2.9) {$d$};
\draw [fill=white] (6,0) circle [radius=0.3];
\node [above] at (6,0.2) {$c$};
\draw [fill=white] (9,0) circle [radius=0.3];
\node [above right] at (9.2,0) {$f$};
\node at (1,-5) {Cyclic path:};
\node at (4,-6) {$a \rightarrow b \rightarrow c \rightarrow d \rightarrow e \rightarrow b \rightarrow c \rightarrow f$};

\node at (15,2) {Processed Network with DP Vertices:};
\draw (12,0) to (15,0);
\draw [decorate,decoration={snake,post length=0.1mm}] (12,0) to (15,0);
\draw (12,0) to [out=300,in=225] (21,-3);
\draw (12,0) to [out=300,in=225] (24,-3);
\draw (15,0) to (18,0);
\draw [decorate,decoration={snake,post length=0.1mm}] (15,0) to (18,0);
\draw (15,0) to [out=270,in=225] (21,-3);
\draw (15,0) to [out=25,in=155] (24,0);
\draw (15,0) to [out=25,in=155] (27,0);
\draw (18,0) to (18,-3);
\draw [decorate,decoration={snake,post length=0.1mm}] (18,0) to (18,-3);
\draw (18,0) to (21,0);
\draw (18,0) to (24,-3);
\draw (18,0) to [out=25,in=155] (27,0);
\draw (18,-3) to (21,-3);
\draw [decorate,decoration={snake,post length=0.1mm}] (18,-3) to (21,-3);
\draw (18,-3) to (24,0);
\draw (21,-3) to (24,-3);
\draw [decorate,decoration={snake,post length=0.1mm}] (21,-3) to (24,-3);
\draw (21,-3) to [out=280,in=290] (27,0);
\draw (21,-3) to (21,0);
\draw (24,-3) to (24,0);
\draw [decorate,decoration={snake,post length=0.1mm}] (24,-3) to (24,0);
\draw (24,-3) to (27,0);
\draw (24,0) to (27,0);
\draw [decorate,decoration={snake,post length=0.1mm}] (24,0) to (27,0);
\draw (24,0) to (21,0);
\draw [fill=white] (12,0) circle [radius=0.3];
\node [above] at (12,0.2) {$a_{0}$};
\draw [fill=white] (15,0) circle [radius=0.3];
\node [above] at (15,0.2) {$b_{0}$};
\draw [fill=white] (18,0) circle [radius=0.3];
\node [above] at (18,0.2) {$c_{0}$};
\draw [fill=white] (18,-3) circle [radius=0.3];
\node [above left] at (17.8,-3) {$d_{0}$};
\draw [fill=white] (21,0) circle [radius=0.3];
\node [above] at (21,0.1) {$d_{1}$};
\draw [fill=white] (21,-3) circle [radius=0.3];
\node [below] at (21,-3.2) {$e_{0}$};
\draw [fill=white] (24,-3) circle [radius=0.3];
\node [right] at (24.2,-2.9) {$b_{1}$};
\draw [fill=white] (24,0) circle [radius=0.3];
\node [above] at (24,0.2) {$c_{1}$};
\draw [fill=white] (27,0) circle [radius=0.3];
\node [above right] at (27.2,0) {$f_{0}$};

\node at (16,-5) {Acyclic path:};
\node at (20,-6) {$a_{0} \rightarrow b_{0} \rightarrow c_{0} \rightarrow d_{0} \rightarrow e_{0} \rightarrow b_{1} \rightarrow c_{1} \rightarrow f_{0}$};

\end{tikzpicture}
\caption{An Example of DP Vertices Pre-processing}
\label{fig:6:prepmap}
\end{figure*}

In this paper, all the unique vertices in the given network $N$ are the \textit{base vertices}; and all the repeated visited vertices in the path $P$ and the duplicate vertices created by DP are \text{sub-vertices}. The relationship between base vertex and its sub-vertex(s) should be public available; otherwise, the path participates cannot recognise the sub-vertices then further reconstruct the exactly original path. Similar to the inheritance in OOP, each sub-vertex inherits the connections from its corresponding base vertex. When inheriting a connection, the sub-vertices have to decide whether such a connection belongs the $G$ or $H$. Figure~\ref{fig:6:prepmap} illustrates the idea of base vertex and sub-vertices. For example, in Figure~\ref{fig:6:prepmap}, because vertex $b$ and $c$ are connected in the network $N$, their sub-vertices: $b_{0}$, $b_{1}$ and $c_{0}$, $c_{1}$ are connected in the processed network $\hat{N}$. Moreover, according to the cyclic path in the original network, $(b_{0}, c_{0}) \in E_{P}$, but $(b_{0}, c_{1}) \notin E_{P}$.

Although such an operation converts a cyclic path to an acyclic one, there is a potential privacy disclosure risk against the adversaries know $m-1$ out of $m$ edges as the background knowledge. In details, since the relationship between the base vertex and sub-vertices is public available, it is easy for the adversaries to learn whether there is a ring or not in the path based on those sub-vertices, then further infer the missing edge. Therefore, to preserve privacy (hide the real ring of the path) against such adversaries, we apply ExpDP to create more sub-vertices (duplicate ones), e.g., vertex $d_{1}$ in Figure~\ref{fig:6:prepmap} is a duplicate sub-vertex created with ExpDP.

To create the duplicate sub-vertices, the main task is about the number of sub-vertices for a given base vertex. To minimise the overall number of vertices after pre-processing, we apply ExpDP to sample a number from a range $[num, \max\{maxNum, 2\}]$ for each base vertex, where $num$ is the number of appearance of a given base vertices in the original path, $maxNum = \max\{num\}$ for all base vertices, e.g., $maxNum = 2$ in the cyclic path in Figure~\ref{fig:6:prepmap} because the maximum appearance of a vertex is 2, comes from both vertex $b$ and $c$. 

To confuse the adversaries whether there is a ring or not in the original path $P$, we want the numbers of sub-vertices of base vertices are indistinguishable. That is, the base vertex with smaller degree (means less chances to be visited repeatedly) would have higher probability to have more duplicate sub-vertices. For simplicity, in this paper, we have the following quality function to measure the number ($n$) of duplicate sub-vertices for base vertex $i$, where $d_{i}$ is the degree of base vertex $i$:
\begin{equation}
\label{eq:6:qv}
q_{v}(n,i) = \max\{n\} - n + \max\{d_{i}\} - d_{i}.
\end{equation}
According to the adversaries' background knowledge, the sensitivity of $q_{v}(n,i)$ is determined by the maximum degree difference when there is an edge in the path missing. That is, $\Delta q_{v} = \max\{|q_{v}(n,i) - q_{v}(n,i^{\prime})|\} = 1, \forall$ base vertex $i$.

Algorithm~\ref{alg:6:dpv} shows how we convert a cyclic path to an acyclic one while adding duplicate sub-vertices to the processed network with ExpDP (a.k.a. differentially private ring removal). Lemma~\ref{lem:6:dpv} proves the differential privacy guarantee of Algorithm~\ref{alg:6:dpv}.
\begin{algorithm}[!ht]
\footnotesize
\caption{Pre-process Differentially Private Vertices (a.k.a. Differentially Private Ring Removal).}
\label{alg:6:dpv}
\SetKwInOut{Input}{Input}
\SetKwInOut{Output}{Output}
\Input{
  A network $N = (V,E)$;\\
  A path $P = \{v_{1}, v_{2}, \dots, v_{n}\}$, $n \geq |V|$;\\
  A quality function: $q_{v}$;\\
  A privacy budgets: $\epsilon_{v}$.
}
\Output{
  A Hamilton network $\hat{N} = (\hat{V}, \hat{E})$;\\
  An acyclic path $\hat{P}$.
}
\BlankLine
$(\hat{V}, \hat{E})$ $\leftarrow$ $(V, E)$\;
$\hat{P}$ $\leftarrow$ $P$\;
$maxDegree$ $\leftarrow$ $\max\{v.degree \text{ in } P, \forall v \in V\}$\;
\For{each $v \in V$}{
  $v.subNum$ $\leftarrow$ Sample a number from $[v.baseNum, maxDegree]$ with $q_{v}$ (Eq.~\eqref{eq:6:qv}) and $\epsilon_{v}$\;
  \For{$i \leftarrow$ 2 to $v.subNum$}{
    $v.sub_{i - 1}$ $\leftarrow$ $v$\;
    \uIf{$i \leq v.baseNum$}{
      replace $i$th $v \in \hat{P}$ with $v.sub_{i - 1}$\;
    }
    \Else{
      $\hat{P}$ $\leftarrow$ $\hat{P}.push(v.sub_{i - 1})$\;
    }
  }
}
\For{each $v \in V$}{
  \For{each $u \in V$, that $e = (v,u) \in E$}{
    \For{each $u.sub_{i} \notin \hat{V}$}{
      $\hat{V}$ $\leftarrow$ $\hat{V} + \{u.sub_{i}\}$\;
      $\hat{E}$ $\leftarrow$ $\hat{E} + \{(v,u.sub_{i})\}$\;
    }
  }
}
\Return: $\hat{N}$, $\hat{P}$\;
\end{algorithm}

\begin{lemma}
\label{lem:6:dpv}
Algorithm~\ref{lem:6:dpv} is $\epsilon_{v}$-differentially private.
\end{lemma}
\begin{proof}
In our vertices pre-processing, we add duplicate vertices to the original network with quality function $q_{v}$ (Equation~\eqref{eq:6:qv}) and privacy budget $\epsilon_{v}$, we have Equation~\eqref{eq:6:dp1} to calculate its overall privacy budget for this step, where $V_{P}$ is the set of real vertices in the path $P$ ($\exists u,v \in V_{P}, u = v$), $V_{m}$ is the set of mixed real and duplicate vertices. Note that, the difference between $V_{P}$ and $V^{\prime}_{P}$ (adversaries') is that we have two pairs of $(v, v^{\prime})$ that $|v.degree - v^{\prime}.degree| = 1$, where $v \in V_{P}$, $v^{\prime} \in V^{\prime}_{P}$, $v = v^{\prime}$. Since we assume that the path is acyclic, we split the set $V_{P}$ to two sub sets: $V_{s}$ and $V_{r}$, where $V_{s}$ contains the vertices appear only once in the path, $V_{r}$ contains the vertices are repeated visited.
\begin{equation}
\label{eq:6:dp1}
\begin{split}
& \frac{\Pr[V_{P} \rightarrow V_{m}]}{\Pr[V^{\prime}_{P} \rightarrow V_{m}]}\\
= & \frac{\Pr[V_{s} \rightarrow V_{m0}] \times \Pr[V_{r} \rightarrow V_{m1}]}{\Pr[V^{\prime}_{s} \rightarrow V_{m0}] \times \Pr[V^{\prime}_{r} \rightarrow V_{m1}]} \\
= & \prod_{v.degree=v^{\prime}.degree}\frac{\Pr[v.subNum = s]}{\Pr[v^{\prime}.subNum = s]} \\
& \times \prod_{v.degree\neq v^{\prime}.degree}\frac{\Pr[v.subNum = s]}{\Pr[v^{\prime}.subNum = s]}\\
= & \prod_{v.degree\neq v^{\prime}.degree}\frac{\Pr[v.subNum = s]}{\Pr[v^{\prime}.subNum = s]}\\
\leq & \exp(\max\{\epsilon_{v}\}) = \exp(\epsilon_{v}).
\end{split}
\end{equation}
The last step of Equation~\eqref{eq:6:dp1} comes from the parallel composition of differential privacy~\cite{MCSHERRY2009}. Therefore, this lemma holds.
\end{proof}

\subsubsection{Differentially Private Edges Pre-processing}
When having a network $N$ contains an acyclic path $P$, according to whether an edge belongs to the path or not, we split the processed network $N = (V, E)$ into two graphs: $G = (V_{G} = V, E_{G} = E \setminus E_{P})$ and $H = P$. Then to fit Rule~\ref{itm:6:fakeedge}, we turn all the non-existent vertex connections in this network to edges either in $G$ or $H$ randomly by ExpDP.

To do so, we firstly create a vertices relationship matrix $\boldsymbol{R}_{|V| \times |V|}$ for the path according to the real topology of the network, where 
\begin{equation}
\label{eq:6:relation}
\left\{
\begin{matrix*}[l]
r_{u,v} = 2, & \text{if $e = (u,v) \in E \setminus E_{P}$;}\\
r_{u,v} = 1, & \text{if $e = (u,v) \in E_{P}$;}\\
r_{u,v} = 0, & \text{if $e=(u,v) \notin E$;}\\
r_{u,v} = -1, & \text{if $u = v$.}
\end{matrix*}
\right.
\end{equation}

Now we shall replace the non-existent edges ($r = 0$) to fake edges either in the path ($r = 1$) or not ($r = 2$) with ExpDP. To conceal any possible missing edge as the adversaries' background knowledge, we want to have more fake edges. Therefore $r = 2$ should have higher quality when implementing ExpDP for adding fake edges. Because the adversaries have different number of edges than us, to measure the quality of $r = 1$ and $r = 2$, in this paper, we have $q_{e}(2) = (|E|-1)/|E|$, $q_{e}(1) = 1/|E|$, namely
\begin{equation}
\label{eq:6:qe}
q_{e}(x) = \frac{|E| - 2}{|E|}x - \frac{|E| - 3}{|E|},
\end{equation}
where $\Delta q_{e} = |q_{e}(1) - q_{e}(2)| = (|E| - 2)/|E| \approx 1$.

Algorithm~\ref{alg:6:dpe} shows our differentially private edges pre-processing. Lemma~\ref{lem:6:dpe} proves the differential privacy guarantee of Algorithm~\ref{alg:6:dpe}.
\begin{algorithm}[!ht]
\footnotesize
\caption{Pre-process Differentially Private Edges.}
\label{alg:6:dpe}
\SetKwInOut{Input}{Input}
\SetKwInOut{Output}{Output}
\Input{
  A network $N = (V,E)$;\\
  A path $P = \{v_{1}, v_{2}, \dots, v_{n}\}$;\\
  A quality function: $q_{e}$;\\
  A privacy budgets: $\epsilon_{e}$.
}
\Output{
  A vertices relational matrix $\boldsymbol{R}_{|V|\times|V|}$.
}
\BlankLine
$\boldsymbol{R}_{|V|\times|V|} \leftarrow$ Traverse $N$ from $v$, values in $\boldsymbol{R}$ determined by Eq.~\eqref{eq:6:relation}\;
\For{each $r_{u,v} = 0$}{
  $r_{u,v}$ $\leftarrow$ 1 or 2 with $q_{e}$ (Eq.~\eqref{eq:6:qe}) and $\epsilon_{e}$\;
}
\Return: $\boldsymbol{R}_{|V|\times|V|}$\;
\end{algorithm}

\begin{lemma}
\label{lem:6:dpe}
Algorithm~\ref{alg:6:dpe} is $(\epsilon_{e} + \ln 2)$-differentially private.
\end{lemma}
\begin{proof}
In our edges pre-processing, because we turn all the non-existent connections in the network $N$ to an edge in either graph $G$ or graph $H$ with quality function $q_{e}$ (Equation~\eqref{eq:6:qe}) and privacy budget $\epsilon_{e}$, we have Equation~\eqref{eq:6:dp2} to calculate its overall privacy budget for this step.
\begin{equation}
\label{eq:6:dp2}
\begin{split}
& \frac{\Pr[N \rightarrow (G,H)]}{\Pr[N^{\prime} \rightarrow (G,H)]}\\
= & \prod_{edgesToG}\frac{\exp(\frac{\epsilon_{e} q_{e}(2)}{2\Delta q_{e}})}{\exp(\frac{\epsilon_{e} q_{e}(2)}{2\Delta q_{e}})} \\
& \times \prod_{edgesToH}\frac{\exp(\frac{\epsilon_{e} q_{e}(1)}{2\Delta q_{e}})}{\exp(\frac{\epsilon_{e} q_{e}(1)}{2\Delta q_{e}})} \times \frac{1}{\Pr[edge( = E \setminus E^{\prime}) ToH]}\\
= & \frac{\exp(\frac{\epsilon_{e} q_{e}(2)}{2\Delta q_{e}}) + \exp(\frac{\epsilon_{e} q_{e}(1)}{2\Delta q_{e}})}{\exp(\frac{\epsilon_{e} q_{e}(1)}{2\Delta q_{e}})}\\
\leq & 2\exp(\epsilon_{e}) = \exp(\epsilon_{e} + \ln 2).
\end{split} 
\end{equation}
Therefore, this lemma holds.
\end{proof}



\subsection{Graph-based Path Producing}
In this section, we have Algorithm~\ref{alg:6:main} to generate our graph-based path from the differentially private matrix $\boldsymbol{R}_{|V|\times|V|}$. In Algorithm~\ref{alg:6:main}, we first call Algorithm~\ref{alg:6:dpv} and Algorithm~\ref{alg:6:dpe} to pre-process a differentially private network, then we insert each vertex of the path into the graph $G_{P}$ one by one by calling Algorithm~\ref{alg:6:insert} in Line 3 of Algorithm~\ref{alg:6:main} to have an initial $G_{P}$ for publishing. To recover the visiting order of the original path with $G_{P}$, we check two groups of vertices: roots and leaves. The group, contains the vertex which is closer to the real first vertex in the path $P$, will serve as the roots in the final output. That is, to fit Rule~\ref{itm:6:root}, we may turn $G_{P}$ upside down to have the final output. This operation is implemented from Line 4 to Line 5 in Algorithm~\ref{alg:6:main}.

\begin{algorithm}[!ht]
\footnotesize
\caption{Publish Graph-based Differentially Private  Path.}
\label{alg:6:main}
\SetKwInOut{Input}{Input}
\SetKwInOut{Output}{Output}
\Input{
  A network $N = (V,E)$;\\
  A path $P = \{v_{1}, v_{2}, \dots, v_{n}\}$;\\
  Two quality functions: $q_{v}$ and $q_{e}$;\\
  Two privacy budgets: $\epsilon_{v}$ and $\epsilon_{e}$.
}
\Output{
  A graph-based path $G_{P}$.
}
\BlankLine
$\hat{N} = (\hat{V}, \hat{E})$, $\hat{P}$ $\leftarrow$ Call Alg.~\ref{alg:6:dpv} with $(N, P, q_{v}, \epsilon_{v})$\;
$\boldsymbol{R}$ $\leftarrow$ Call Alg.~\ref{alg:6:dpe} with $(\hat{N}, \hat{P}, q_{e}, \epsilon_{e})$\;
$(\hat{V},\hat{E})$ $\leftarrow$ Call Alg.~\ref{alg:6:insert} with $(\hat{P}, \boldsymbol{R})$\;
\If{$u \in \hat{V}_{leaves}$ is closer to $v_{1}$}{
  $\hat{V}$ $\leftarrow$ reverse the order of layers in $\hat{V}$\;
}
\Return: $G_{P}$ $\leftarrow$ $(\hat{V},\hat{E})$\;
\end{algorithm}

According to the sequential composition of differential privacy~\cite{MCSHERRY2009}, based on Lemma~\ref{lem:6:dpv} and Lemma~\ref{lem:6:dpe}, we have Theorem~\ref{thm:6:privacy} to study the overall privacy budget of differential privacy of Algorithm~\ref{alg:6:main}.
\begin{theorem}
\label{thm:6:privacy}
Algorithm~\ref{alg:6:main} is $(\epsilon_{v} + \epsilon_{e} + \ln 2)$-differentially private.
\end{theorem}

Furthermore, Lemma~\ref{lem:6:roots} and Proposition~\ref{prop:6:firstvertex} study the probability of recovering the correct the path $P$ from the published $G_{P}$.
\begin{lemma}
\label{lem:6:roots}
There are at least two roots in the graph-based path, where two of the roots share an edge of the path. 
\end{lemma}
\begin{proof}
Assume there is only one root, $v$, in the graph-based path, then $v$'s neighbour vertex in the path, $u$, will be in the same branch as $v$. This is contradict to our first rule for this graph-based path, namely, $v$ and $u$ should be placed into different branches. That is, at least both $u$ and $v$ are the roots of the graph-based path. Therefore, this corollary holds.
\end{proof}

\begin{proposition}
\label{prop:6:firstvertex}
Given a path $P = \{v_{1}, v_{2}, \dots, v_{n}\}$, our graph-based path $G_{P}$ (without DP injection) has at least $1/2$ probability to confirm the visiting order of $P$ in expected case.
\end{proposition}
\begin{proof}
According to Lemma~\ref{lem:6:roots}, we assume there are $k \geq 4$ vertices are roots and leaves in $G_{P}$. There is only one case that $G_{P}$ cannot confirm the first node of $P$: Among the $k$ vertices, vertex $v$, which is the closest one to the first node, and vertex $u$, which is the closest one to the last node, are both roots. In expected cases, the probability of such case in all possible graph $G$ is $2/{k\choose 2} = \frac{4}{k(k-1)}$ because the roots set and the leaves set are exchangeable. Moreover, $u$ should be closer to the last node than $v$ to the first node. We assume it has $Pr \in [0,1]$ probability to have such case. Then $G_{P}$ has $1 - Pr \times \frac{4}{k(k-1)}$ probability to confirm the real visiting order of $P$, which is at least $2/3$. In addition, consider a special network with $n$ vertices and $n-1$ edges. Give a $P$ from such network, because both first node and last node will be the root in the graph-based path, we have $1/2$ probability to confirm the real first node. Therefore, we have probability $1/2 = \min\{1/2, 2/3\}$ to confirm the visiting order of of $P$, so that proposition holds.
\end{proof}

\begin{algorithm}[!ht]
\footnotesize
\caption{Insert Vertex.}
\label{alg:6:insert}
\SetKwInOut{Input}{Input}
\SetKwInOut{Output}{Output}
\Input{
  A path $P$;\\
  A vertices relationship matrix $\boldsymbol{R}$.
}
\Output{
  A vertex superset $\hat{V}$ and an edge set $\hat{E}$
}
\BlankLine
Init: An empty stack $S$, $S_{bak} \leftarrow S$, $layer.top \leftarrow 0$, $layer.bottom \leftarrow 0$, $u \leftarrow P.dequeue()$, $u.layer \leftarrow 0$\;
$S.push(u)$\;
$v \leftarrow P.pop()$\;
\While{$P \neq \emptyset$}{
  \If{$v.visit = 0$}{
    \For{$i$ $\leftarrow$ $layer.top-1$ to $layer.bottom+1$}{
      \If{Call Alg.~\ref{alg:6:good} with $(S,v,\boldsymbol{R},i)$}{
        $v.layerRange$ $\leftarrow$ $v.layerRange + \{i\}$\;
      }
    }
  }
  \uIf{$v.layerRange \neq \emptyset$}{
    $v.layer$ $\leftarrow$ $\max\{v.layerRange\}$\;
    $layer.top$ $\leftarrow$ $\min_{\forall x \in S}\{x.layer, v.layer\}$\;
    $layer.bottom$ $\leftarrow$ $\max_{\forall x \in S}\{x.layer, v.level\}$\;
    $v.visit \leftarrow 1$, $S.push(v)$\;
    $v \leftarrow P.pop()$, $v.visit$ $\leftarrow$ 0\;
    \If{$|S| > |S_{bak}|$}{
      $(S_{bak},\boldsymbol{R}_{bak},P_{bak}, v_{bak})$ $\leftarrow$ $(S,\boldsymbol{R},P, v)$\;
    }
  }
  \Else{
    $P.push(v)$\;
    $v \leftarrow S.pop()$, $v.layerRange$ $\leftarrow$ $v.layerRange - \{v.layer\}$, $v.layer$ $\leftarrow$ $nil$\;
    \If{$v = u$}{
      $(S,\boldsymbol{R}) \leftarrow$ Call Alg.~\ref{alg:6:split} with $(S_{bak}, \boldsymbol{R}_{bak}, v_{bak})$\;
      $P \leftarrow P_{bak}$, $v \leftarrow P.pop()$, $v.visit \leftarrow 0$\;
    }
  }
}
\For{each $v \in S$}{
  \If{$r_{v,x} = 2, \forall x \in S$}{
    $\hat{E} \leftarrow e=(v,x)$\;
  }
  $V_{v.layer}$ $\leftarrow$ $V_{v.layer} + \{v\}$\;
}
$\hat{E} \leftarrow$ Remove duplicate edges\;
$\hat{V} \leftarrow \{V_{i}\}, \forall i \in [layer.top, layer.bottom]$\;
\Return: $(\hat{V}, \hat{E})$\;
\end{algorithm}

Algorithm~\ref{alg:6:insert} is the principal part for producing the graph-based path $G_{P}$. Generally, in Algorithm~\ref{alg:6:insert}, we maintain and track the status of each vertex in the process of inserting to ensure the vertices are correctly placed in the graph-based path. Clearly, because we hide the exact information of the path in the topology of the graph-based path, the layers in $G_{P}$, where we place the vertices, play an important role for us. So we prepare a possible layers set for each vertex against the present graph at the moment when placing the vertex into $G_{P}$ (Line 5 to Line 8 in Algorithm~\ref{alg:6:insert}). To ensure the first node of the path could be placed as the root in the graph-based path, we always use the maximum layer from the set for inserting the vertex (Line 10 to Line 14 in Algorithm~\ref{alg:6:insert}). 

To compute such a possible layers set, we exclude all the bad layers which introduce contradiction to our graph building rules in Algorithm~\ref{alg:6:good}. Claim~\ref{clm:6:trans} studies the way to navigate the bad layers when inserting a vertex $v$ to $G_{P}$.

\begin{algorithm}[!ht]
\footnotesize
\caption{Test Insertion Layer.}
\label{alg:6:good}
\SetKwInOut{Input}{Input}
\SetKwInOut{Output}{Output}
\Input{
  A stack $S$;\\
  A vertex $u$;\\
  A vertices relationship matrix $\boldsymbol{R}$;\\
  A given layer $n$.
}
\Output{
  A boolean status.
}
\BlankLine
\If{$r_{u's\ parent,u's\ child} = 1$}{
  \Return: False\;
}
\For{each $v \in S$}{
  \If{$v.layer = n$ AND $r_{u,v} = 2$}{
    \Return: False\;
  }
  \If{$v.layer > n$ AND $r_{u,v} = 2$ AND $r_{u,v's\ child} = 1$}{
    \Return: False\;
  }
  \If{$v.layer < n$ AND $r_{u,v} = 2$ AND $r_{u, v's\ parent} = 1$}{
    \Return: False\;
  }
}
\Return: True\;
\end{algorithm}

\begin{claim}
\label{clm:6:trans}
The connectivity is transitive in a branch in $G_{P}$.
\end{claim}

Claim~\ref{clm:6:trans} shows an important feature of our graph-based path: all the vertices in one branch should be connected to each other. That is, when inserting $v$ to $G_{P}$, a bad layer for $v$ breaks either our Rule~\ref{itm:6:samebranch} for $G_{P}$ or such transitive relation in Claim~\ref{clm:6:trans}. The upper half of Figure~\ref{fig:6:trans} shows a counterexample to Claim~\ref{clm:6:trans}. Given $u, x \in G_{P}$, $r_{u,x} = 2$, when inserting $v \in P$ to $G_{P}$ where $r_{v,u} = r_{v,x} = 1$, there is no good layer for $v$, if $\exists y \in G_{P}$ that $r_{y,v} = r_{y,u} = r_{y,x} = 2$ and $y$'s layer in $G_{P}$ is between $u$'s and $x$'s. This case indicates that there may be bad placement(s) from the existing vertices which are already placed in the present graph. To overcome the bad placement(s), we record each inserted vertex in a stack. When we have a vertex with an empty possible layers set, the stack enables us to backtrack and re-insert the recently inserted vertices until we can successfully insert all the vertices into the graph. This backtracking operation is implemented from Line 18 to Line 19 in Algorithm~\ref{alg:6:insert}. The lower half of Figure~\ref{fig:6:trans} illustrates an example for addressing this issue.
\begin{figure}[!th]
\centering
\begin{tikzpicture}[xscale=0.5, yscale=0.5]
\draw (0,0) to (0,-1.7);
\draw (0,-2.3) to (0,-4);

\draw (2,0) to (2,-1.7);

\draw (4,0) to (4,-1.7);
\draw (4,-2.3) to (4,-4);
\draw [dashed] (6,0) to (4.2,-1.8);
\draw [dashed] (6,-2) to (4.3,-2);
\draw [dashed] (6,-4) to (4.2,-2.2);

\draw (0,-7.3) to (0,-8.7);
\draw (0,-9.3) to (0,-11);

\draw (2,-7) to (2,-8.7);

\draw (4,-7.3) to (4,-8.7);
\draw (4,-9.3) to (4,-11);
\draw (4.2,-7.2) to (6,-9) ;

\draw [fill=white] (0,0) circle [radius=0.3];
\node [above left] at (0,0) {$u$};
\draw [pattern=north east lines] (0,-2) circle [radius=0.3];
\node [above left] at (0,-2) {$y$};
\draw [fill=white] (0,-4) circle [radius=0.3];
\node [above left] at (0,-4) {$x$};

\draw [fill=white] (2,0) circle [radius=0.3];
\node [above left] at (2,0) {$v$};
\draw [pattern=north east lines] (2,-2) circle [radius=0.3];
\node [above left] at (2,-2) {$y$};

\draw [fill=white] (4,0) circle [radius=0.3];
\node [above left] at (4,0) {$u$};
\draw [pattern=north east lines] (4,-2) circle [radius=0.3];
\node [above left] at (4,-2) {$y$};
\draw [fill=white] (4,-4) circle [radius=0.3];
\node [above left] at (4,-4) {$x$};
\draw [fill=white] (6,0) circle [radius=0.3];
\node [above left] at (6,0) {$v$};
\draw [fill=white] (6,-2) circle [radius=0.3];
\node [above left] at (6,-2) {$v$};
\draw [fill=white] (6,-4) circle [radius=0.3];
\node [above] at (6,-3.8) {$v$};

\draw [fill=white] (0,-9) circle [radius=0.3];
\node [above left] at (0,-9) {$u$};
\draw [pattern=north east lines] (0,-7) circle [radius=0.3];
\node [above left] at (0,-7) {$y$};
\draw [fill=white] (0,-11) circle [radius=0.3];
\node [above left] at (0,-11) {$x$};

\draw [fill=white] (2,-7) circle [radius=0.3];
\node [above left] at (2,-7) {$v$};
\draw [pattern=north east lines] (2,-9) circle [radius=0.3];
\node [above left] at (2,-9) {$y$};

\draw [fill=white] (4,-9) circle [radius=0.3];
\node [above left] at (4,-9) {$u$};
\draw [pattern=north east lines] (4,-7) circle [radius=0.3];
\node [above left] at (4,-7) {$y$};
\draw [fill=white] (4,-11) circle [radius=0.3];
\node [above left] at (4,-11) {$x$};
\draw [fill=white] (6,-9) circle [radius=0.3];
\node [above] at (6,-8.8) {$v$};
\node [above] at (-0.1,0.5) {$G_{P}$:};
\node at (1,-2) {$+$};
\node at (3,-2) {$\rightarrow$};
\node [right] at (6.2,0) {(bad: $v$ connects to $x$)};
\node [right] at (6.2,-2) {(bad: $v$ and $y$ not in a branch)};
\node [right] at (6.2,-4) {(bad: $v$ connects to $u$)};

\node [above] at (-0.1,-6.4) {Adjusted $G_{P}$:};
\node at (1,-9) {$+$};
\node at (3,-9) {$\rightarrow$};
\node [right] at (6.2,-9) {(good)};
\end{tikzpicture}
\caption{An Example of Claim~\ref{clm:6:trans}}
\label{fig:6:trans}
\end{figure}

Based on Claim~\ref{clm:6:trans}, we have Theorem~\ref{thm:6:exist} which studies the existence of our graph-based path. 
\begin{theorem}
\label{thm:6:exist}
The graph-based path exists, unless $\exists u \in P$, that $u$ breaks the transitive relation in all possible topology of $G_{P}$.
\end{theorem}

According to Theorem~\ref{thm:6:exist}, it is possible to have a network which cannot deliver a graph-based path $G_{P}$. Namely, to insert a vertex $u$ to a $G_{P}$, after trying all the possible layers of each vertex in the present $G_{P}$ by backtracking the stack in Algorithm~\ref{alg:6:insert}, we cannot find a good layer to insert $u$. The upper half of Figure~\ref{fig:6:split} shows an example of bad $G_{P}$ by Theorem~\ref{thm:6:exist} where the present $G_{P}$ is the only possible topology with the given vertices $a$, $b$, $c$, and $d$. As we can see, there is no way to appropriately insert $e$ to the present $G_{P}$.

\begin{figure}[!th]
\centering
\begin{tikzpicture}[xscale=0.5, yscale=0.5]
\draw (0,0) to (0,-3);
\draw (3,0) to (3,-3);
\draw (3,0) to (0,-3);
\draw (6.5,0) to (5,-3);
\draw (6.5,0) to (8,-3);

\draw (10,0) to (10,-3);
\draw (13,0) to (13,-3);
\draw (13,0) to (10,-3);
\draw (10,0) to (13,-3);

\draw (5,-6) to (5,-9);
\draw (8,-6) to (8,-9);
\draw (8,-6) to (5,-9);
\draw (2,-9) to (5,-6);
\draw (11,-6) to (8,-9);
\draw [fill=white] (0,0) circle [radius=0.3];
\node [above left] at (0,0) {$a$};
\draw [fill=white] (3,0) circle [radius=0.3];
\node [above left] at (3,0) {$b$};
\draw [fill=white] (0,-3) circle [radius=0.3];
\node [above left] at (0,-3) {$c$};
\draw [fill=white] (3,-3) circle [radius=0.3];
\node [above left] at (3,-3) {$d$};
\draw [fill=white] (5,-3) circle [radius=0.3];
\node [above] at (5,-2.7) {$a$};
\draw [fill=white] (8,-3) circle [radius=0.3];
\node [above] at (8,-2.7) {$d$};
\draw [fill=white] (6.5,0) circle [radius=0.3];
\node [above left] at (6.5,0) {$e$};

\draw [fill=white] (10,0) circle [radius=0.3];
\node [above left] at (10,0) {$a$};
\draw [fill=white] (13,0) circle [radius=0.3];
\node [above left] at (13,0) {$b$};
\draw [fill=white] (10,-3) circle [radius=0.3];
\node [above left] at (10,-3) {$c$};
\draw [fill=white] (13,-3) circle [radius=0.3];
\node [above right] at (13,-3) {$d$};
\draw [fill=white] (11,-1) circle [radius=0.3];
\node [above] at (11,-0.7) {$e$};

\draw [fill=white] (5,-6) circle [radius=0.3];
\node [above left] at (5,-6) {$a$};
\draw [fill=white] (8,-6) circle [radius=0.3];
\node [above left] at (8,-6) {$b$};
\draw [fill=white] (5,-9) circle [radius=0.3];
\node [above left] at (5,-9) {$c$};
\draw [fill=white] (8,-9) circle [radius=0.3];
\node [above left] at (8,-9) {$d$};
\draw [fill=white] (2,-9) circle [radius=0.3];
\node [above left] at (2,-9) {$e$};
\draw [fill=white] (11,-6) circle [radius=0.3];
\node [above left] at (11,-6) {$e$};

\node at (0,1.5) {$G_{P}$:};
\node at (6,1.5) {Insert $e$};
\node at (4,-1.5) {$+$};
\node at (9,-1.5) {$\rightarrow$};
\node [right] at (9,2) {$G_{P}$};
\node [right] at (9,1.2) {with conflict:};
\node at (6.5,-4) {$\downarrow$};
\node at (2.5,-5) {$G_{P}$ without conflict:};

\end{tikzpicture}
\caption{An Example of Resolving Conflicts in $G_{P}$}
\label{fig:6:split}
\end{figure}
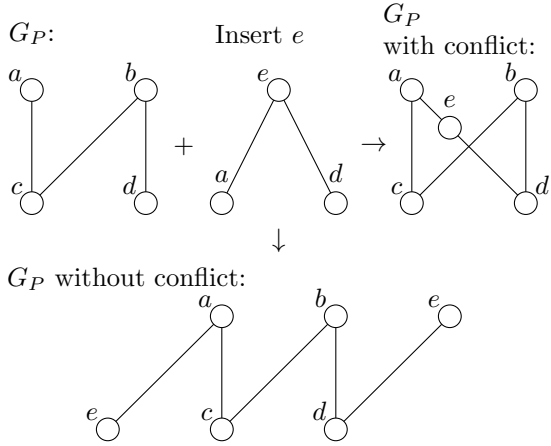

The reason behind this conflict when inserting a vertex is straightforward. When a vertex $u$ has multiple connections in the network, according to Rule~\ref{itm:6:notsamebranch}, Rule~\ref{itm:6:samebranch}, and Claim~\ref{clm:6:trans}, the topology of the graph-based path may not satisfy all the transitive relation relevant to $u$ at the same time. To address this problem, we propose Algorithm~\ref{alg:6:split} to eliminate such a conflict. Generally, the main idea of Algorithm~\ref{alg:6:split} is splitting the vertices, which break the transitive relation in branches in the present $G_{P}$ when $u$ comes, to several pieces according to its degree and connectivity in the network. As a result, each piece takes less connections from the original vertex. This way, it is easier to satisfy the connectivity of the conflict vertex. In Algorithm~\ref{alg:6:split}, when splitting a vertex $v$, we have $v.connections = \bigcap_{i} v_{i}.connections$, where $i \leq v.degree$. The lower half of Figure~\ref{fig:6:split} illustrates that after splitting the vertex $e$, we resolve the conflict by splitting vertex $e$. The new $G_{P}$ satisfies all our rules for a graph-based path.

\begin{algorithm}[!ht]
\footnotesize
\caption{Resolve Inserting Conflicts.}
\label{alg:6:split}
\SetKwInOut{Input}{Input}
\SetKwInOut{Output}{Output}
\Input{
  A stack $S$;\\
  A vertices relationship matrix $\boldsymbol{R}$;\\
  A vertex $v_{0}$.
}
\Output{
  A stack $S$;\\
  A vertices relationship matrix $\boldsymbol{R}$.
}
\BlankLine
$S_{conflict}$ $\leftarrow$ $\emptyset$\;
$layer.top \leftarrow \min_{\forall u \in S}\{u.layer\}$\;
$layer.bottom \leftarrow \max_{\forall u \in S}\{u.layer\}$\;
\For{each $u \in S$}{
  $S_{conflict}.push($the conflicting $u$ for $v)$\;
  $S.remove(u)$\;
}
\If{$S_{conflict} = \emptyset$}{
  $S_{conflict}.push(v_{0})$\;
}
\For{each $u \in S_{conflict}$}{
  $S_{temp}$ $\leftarrow$ $\emptyset$\;
  \For{$l \in [layer.top - 1, layer.bottom + 1]$}{
    $S_{l}$ $\leftarrow$ Insert the suitable $v \in S$ if $u.layer = l$\;
    $S_{temp}.push(S_{l})$\;
  }
  $S_{temp}$ $\leftarrow$ Solve a set cover for $u$ and $S_{temp}$\;
  Split $u$ according to $S_{temp}$\;
  $(S, \boldsymbol{R})$ $\leftarrow$ Insert the split $u$s to $S$ and $\boldsymbol{R}$\;
}
\If{$v_{0} \notin S$}{
  $S$ $\leftarrow$ Insert $v$ to $S$ as Alg.~\ref{alg:6:insert}\;
}
\Return: $S$ and $\boldsymbol{R}$\;
\end{algorithm}

Note that, the splitting operation in Algorithm~\ref{alg:6:split} is different to the differentially private vertices pre-processing in Algorithm~\ref{alg:6:dpv}. Because in Algorithm~\ref{alg:6:dpv}, the sub-vertices are duplicates of the base vertex but have different connections according to the positions of the duplicates, i.e., some of the sub-vertices are fake ones; while in Algorithm~\ref{alg:6:split}, after splitting the base vertex, each piece of the base vertex inherits partial connections from its base vertex, i.e., those pieces of vertices are not fake ones. Theorem~\ref{thm:6:split} proves that this splitting operation ensures us always having a graph-based path from any input network and path.

\begin{theorem}
\label{thm:6:split}
Algorithm~\ref{alg:6:split} guarantees to produce a graph which satisfies our rules for the graph-based path.
\end{theorem}
\begin{proof}
When we split a vertex, we do not create new edges, i.e., new connections between vertices, so this splitting strategy keeps all the connectivities in $G$ (show in Figure~\ref{fig:6:example}). Furthermore, by splitting the conflict vertex $v$, we avoid satisfying $v$'s connectivity all together by one vertex but multiple new vertices separately. Especially, in the worst case, when inserting a vertex $v$ for all vertices, we can always split it according to its degree in $G$. This way, we will finally have a forest $G_{P}$ where each single edge in $G$ served as a tree. Therefore, such graph follows all our rules for building our graph-based path for recovering the original path.
\end{proof}

Finally, we study the depth of $G_{P}$ in Proposition~\ref{prop:6:height} which assists the proof of Corollary~\ref{cor:6:recover} for the time complexity for each path participate recovering its status in the path.
\begin{proposition}
\label{prop:6:height}
The depth of the graph-based path, which is the length of the longest branch, is at most $|V_{G}|/2$.
\end{proposition}
\begin{proof}
According to the way we build the graph-based path, the longest branch comes from the maximum-sized complete sub-graph in $G$. Because we construct $G$ by removing the edges of path from (and adding the non-existent edges into) the original network, in the worst case, there are at most $|V_{G}|/2$ vertices could form such a complete sub-graph in $G$. Therefore, this proposition holds.
\end{proof}

\begin{corollary}
\label{cor:6:recover}
It takes $\mathcal{O}(|V_{G}|)$ time complexity for a path participate to recover its status in the path from our graph-based path ($G_{P}$ in Figure~\ref{fig:6:example}).
\end{corollary}
\begin{proof}
For path participate, to recover its status from $G$, it takes $\mathcal{O}(|V_{G}|+|E_{G}|)$ time complexity to traverse all the edges connected with its one end in $G$. In the worst case, $\mathcal{O}(|E_{G}|)$ is equivalent to $\mathcal{O}(|V_{G}|^{2})$. However, according to Proposition~\ref{prop:6:height}, each vertex in $G_{P}$ has at most $\mathcal{O}{|V_{G}|}$ edges. Then, for path participate, it takes $\mathcal{O}(|V_{G}| + |V_{G}|)$ to search all the edges connected with its one end in $G_{P}$. Therefore, $\mathcal{O}(|V_{G}|)$ time complexity to confirm whether an edge is in the path or not.
\end{proof}

\section{Experimental Evaluation}
\label{sec:6:exp}
In this section, we show our experimental evaluation against our proposed algorithm, including the dataset we used, experimental configurations, and the experimental results.

\subsection{Dataset and Configurations}
According to the inputs of our algorithm, we need a dataset contains a map where all vertices form the path. To evaluate the performance of our algorithm over different sizes of graphs, in this paper, similar to \cite{FerreiraN2013,AlahiA2016}, we generate and use a synthetic dataset: undirected map with different sizes and densities, which contains the path (either acyclic or cyclic) for publishing. We produce the undirected map with a given number of vertices $|V|$ and a given number of edges $|E|$ in three steps. First, we sample all the vertices from a two-dimensional plane uniformly. Second, we draw a path for publishing to connect all the vertices in the map. Third, all the edges not in the path but in the map are randomly generated with uniform distribution.

Because our differentially private algorithm is randomised, according to the law of large numbers~\cite{WIKI2019}, we report the expected algorithm performance of each aspects by calculating the average results of running the algorithm 1000 times. To simulate the real-world scenario, the size of a map/path in our experiments is no more than 10 vertices. Accordingly, the number of edges of the synthetic map is in the range from $|V| - 1$ (a map only has a Hamilton path for publishing) to $\frac{1}{2}|V|(|V| - 1)$ (a complete map). So, in our experiments, we have overall 122 different sizes of maps as inputs.

The experiments are programmed in Python 3.7, plotted with MATLAB R2019b, on a node, with 16 cores and 64GB memory, of a super computer (the Phoenix HPC service at the University of Adelaide). 

\subsection{Experimental Results}
In this section, we evaluate the performance of our algorithm from the following aspects: run time of our algorithm with differential privacy on either vertices or edges; output quality of the random maps without both differential privacy and the splitting strategy for the graph-based path; quality of the differentially private output. Since our graph-based path can definitely confirm whether an edge belongs to the original path, in this section, we evaluate the output (graph-based path) quality by testing whether the output could confirm the visiting order of the original path as showed in Section~\ref{sec:6:overview}. In the figures in this section, each point on a curve stands for a specific map with given number of vertices and edges. 

\subsubsection{Run Time}
\begin{figure*}[!ht]
	\subfloat[DP on vertex ($\epsilon_{v} = 0.5$)]{%
		\includegraphics[width=0.24\textwidth]{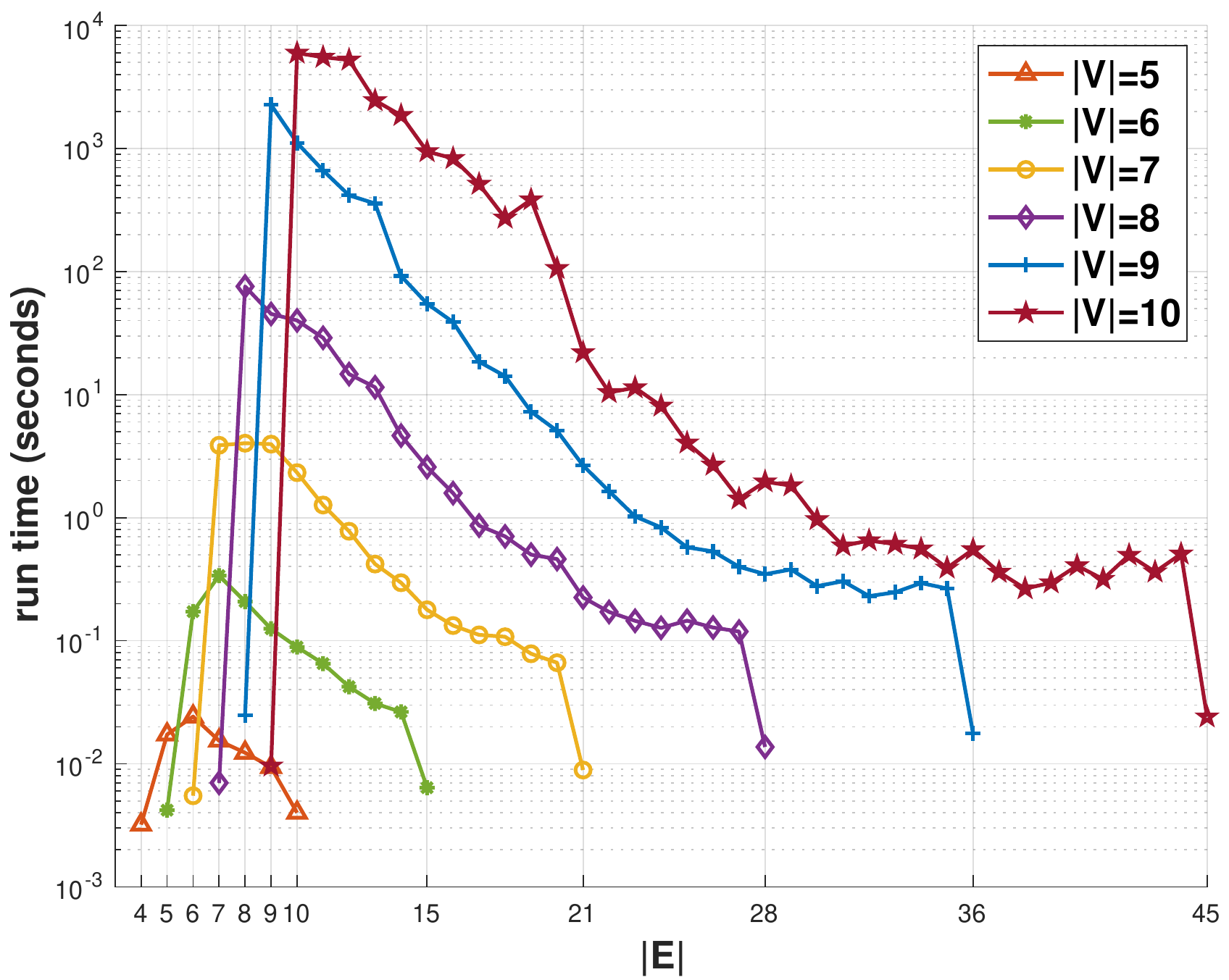}
		\label{subfig:6:run_time_v0}
	}
	\hfill
	\subfloat[DP on vertex ($\epsilon_{v} = 1$)]{%
		\includegraphics[width=0.24\textwidth]{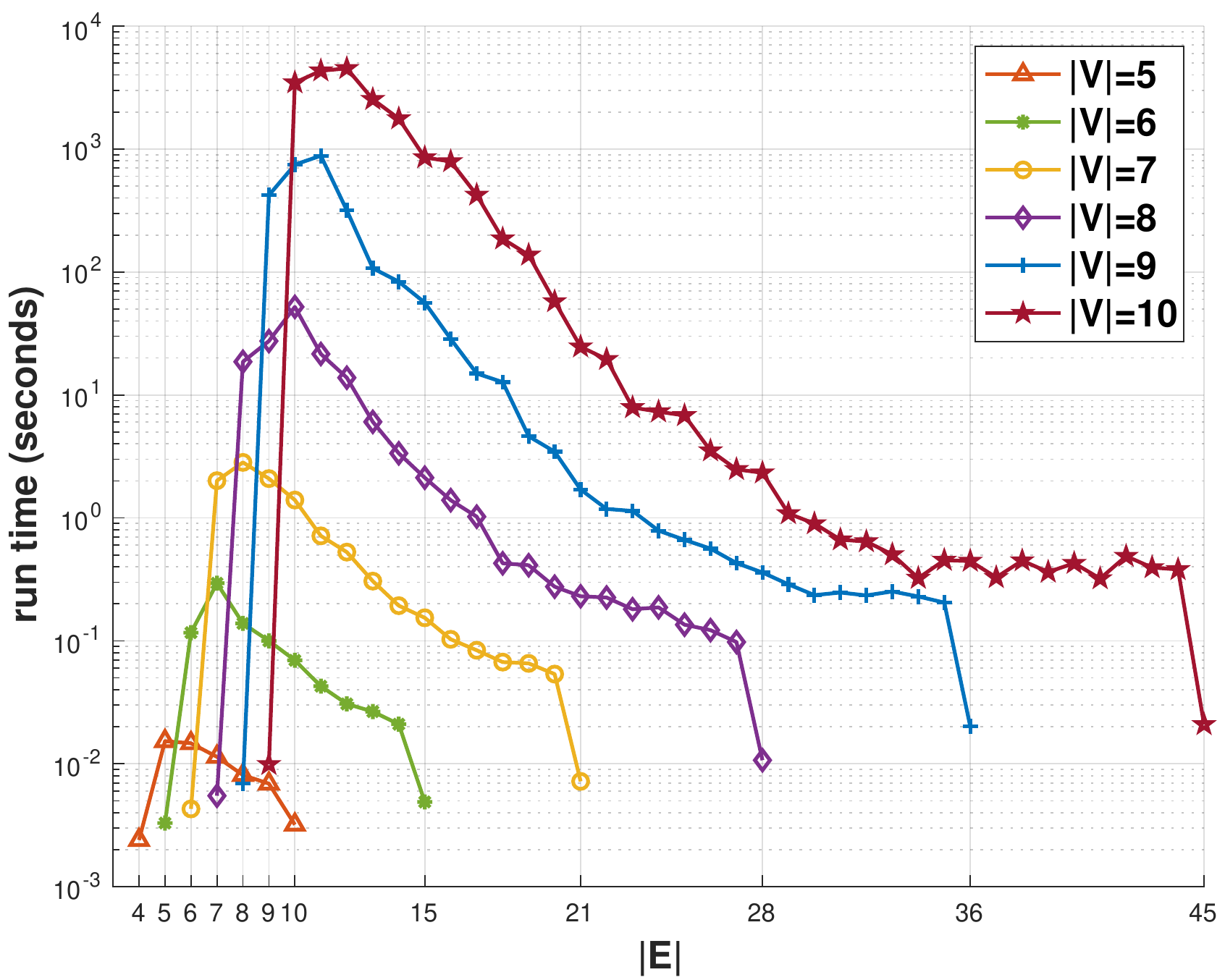}
		\label{subfig:6:run_time_v1}
	}
	\hfill
	\subfloat[DP on edge ($\epsilon_{e} = 0.5$)]{%
		\includegraphics[width=0.24\textwidth]{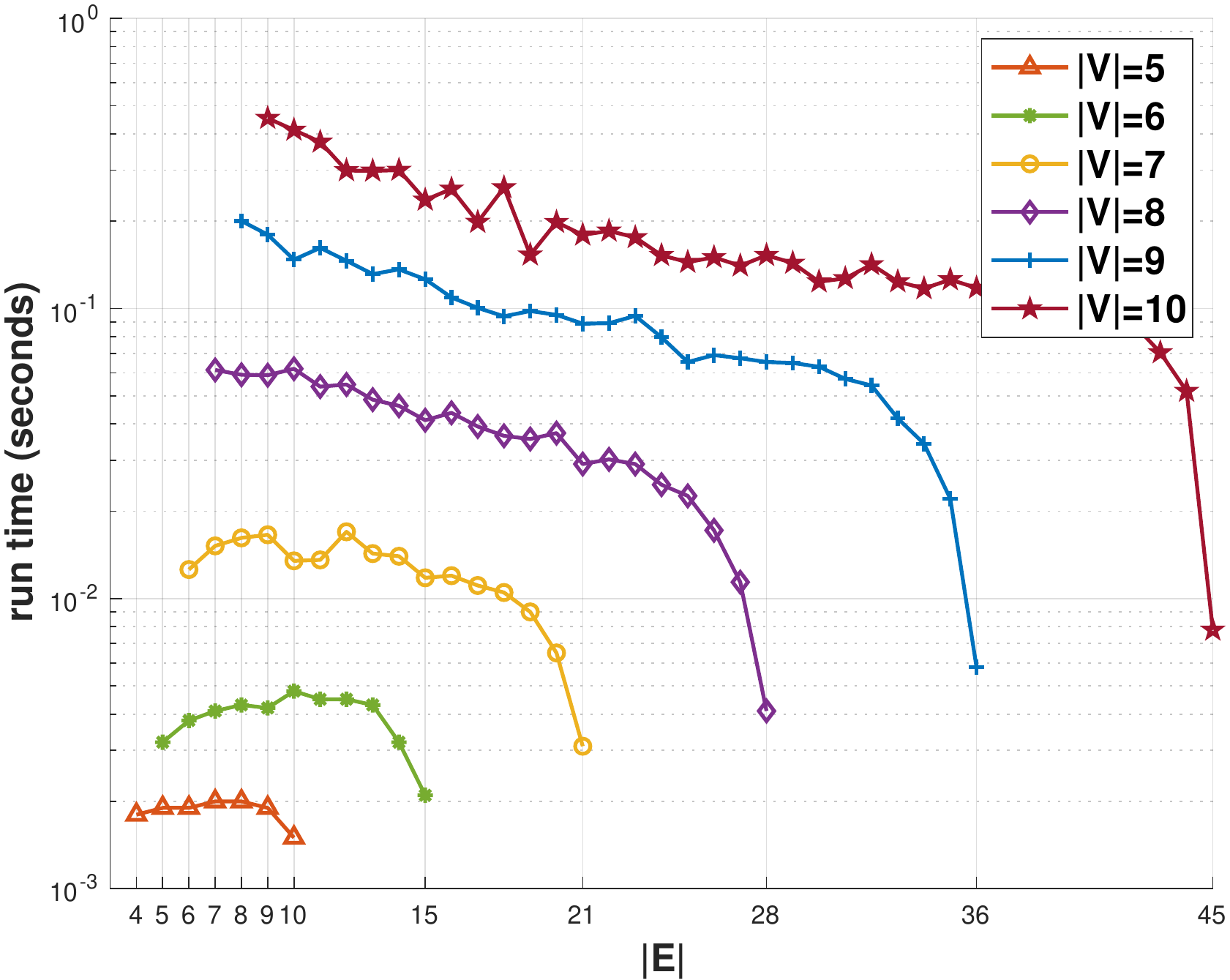}
		\label{subfig:6:run_time_e0}
	}
	\hfill
	\subfloat[DP on edge ($\epsilon_{e} = 1$)]{%
		\includegraphics[width=0.24\textwidth]{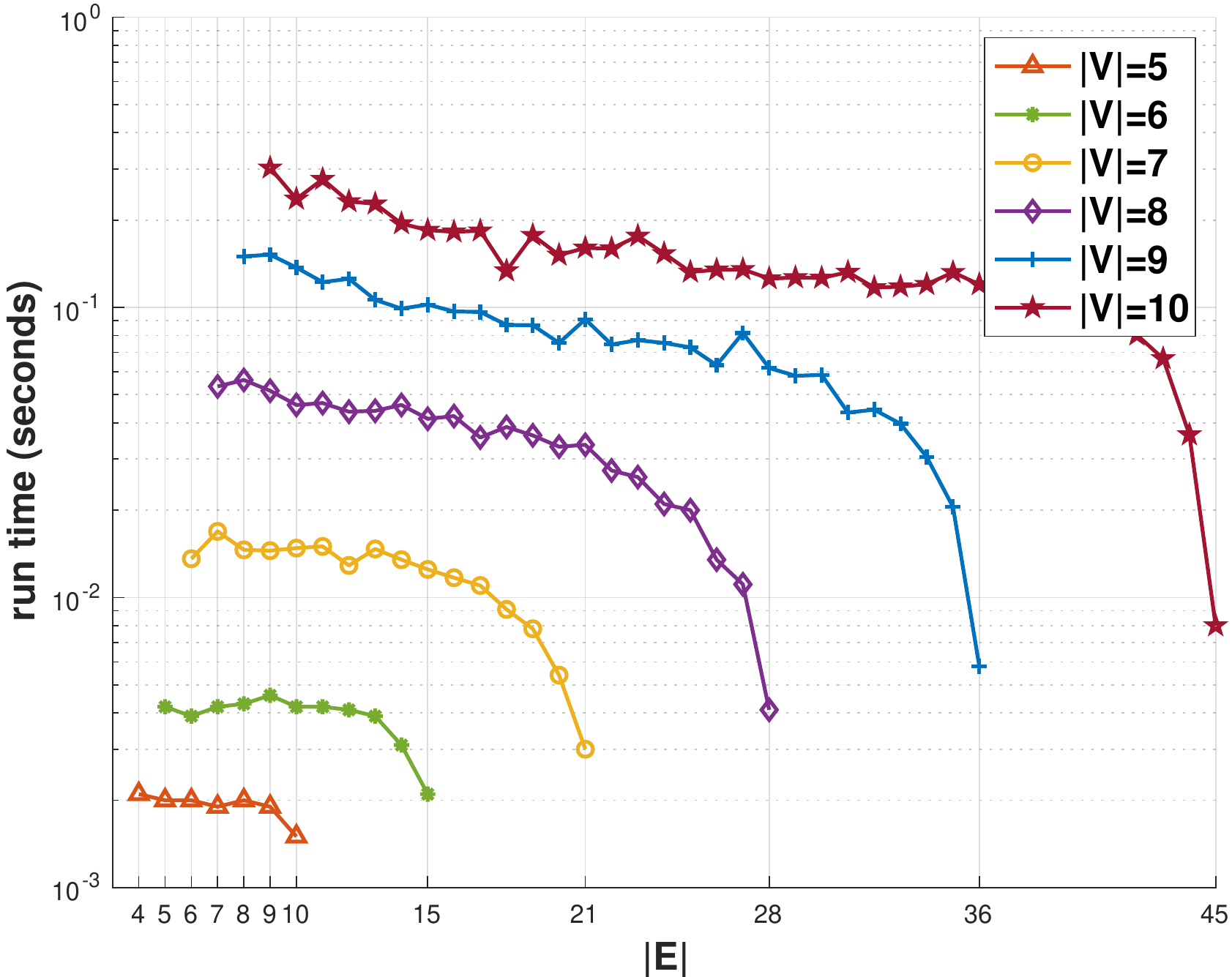}
		\label{subfig:6:run_time_e1}
	}
	\caption{Run Time (Acyclic Path).}
	\label{fig:6:run_time}
\end{figure*}
Figure~\ref{fig:6:run_time} depicts the run time of our algorithm with acyclic paths over different sizes of the maps and different settings of differential privacy on either vertices or edges. It is clear that when increasing the number of vertices, the running time of our algorithm increases significantly. However, from Figure~\ref{subfig:6:run_time_e0} and Figure~\ref{subfig:6:run_time_e1}, the increasing number of edges results in a decreasing run time. The reason for it is that with a given number of vertices and a given $\epsilon_{e}$, when increasing the number of edges, less spaces are available for adding fake edges, then the number of possible topologies of the map is decreasing, so it will take less time for our algorithm to find a suitable graph-based path. 

According to Equation~\eqref{eq:6:qv}, more vertices are injected with a smaller privacy budget on vertices $\epsilon_{v}$, so that the run time in Figure~\ref{subfig:6:run_time_v0} ($\epsilon_{v} = 0.5$) is much higher than the run time in Figure~\ref{subfig:6:run_time_v1} ($\epsilon_{v} = 1$). Furthermore, based on Equation~\eqref{eq:6:qe}, more edges are injected to the map ($G$ in Figure~\ref{fig:6:example}) with a greater privacy budget on edges $\epsilon_{e}$, then the run time in Figure~\ref{subfig:6:run_time_e0} ($\epsilon_{e} = 0.5$) is much higher than the run time in Figure~\ref{subfig:6:run_time_e1} ($\epsilon_{e} = 1$).

\subsubsection{Outputs Quality over Random Maps without DP}
\begin{figure*}[!ht]
	\subfloat[Percentage of Usable Map]{%
		\includegraphics[width=0.32\textwidth]{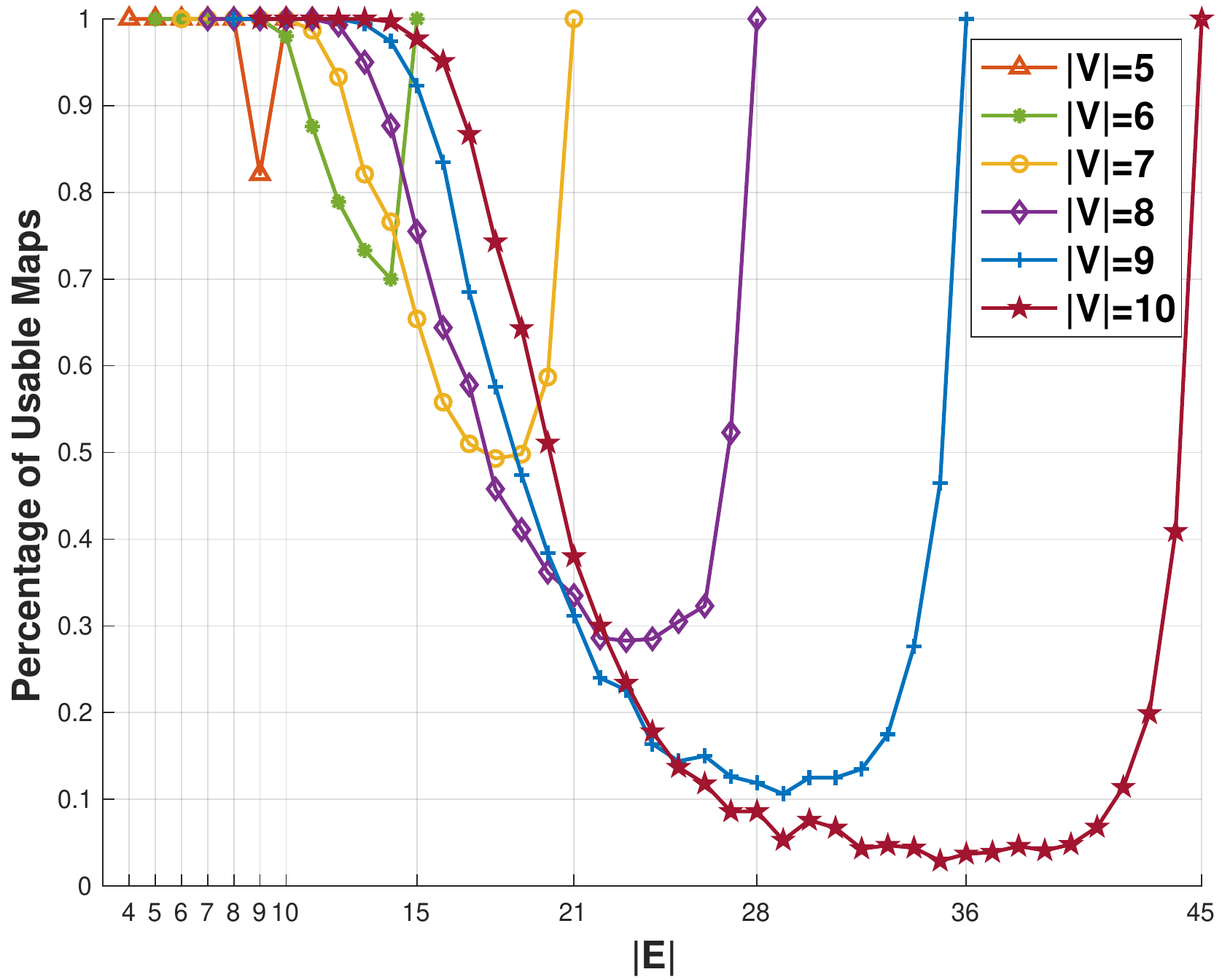}
		\label{subfig:6:good_map_no_dp}
	}
	\hfill
	\subfloat[Percentage of Good Output with Usable Map]{%
		\includegraphics[width=0.32\textwidth]{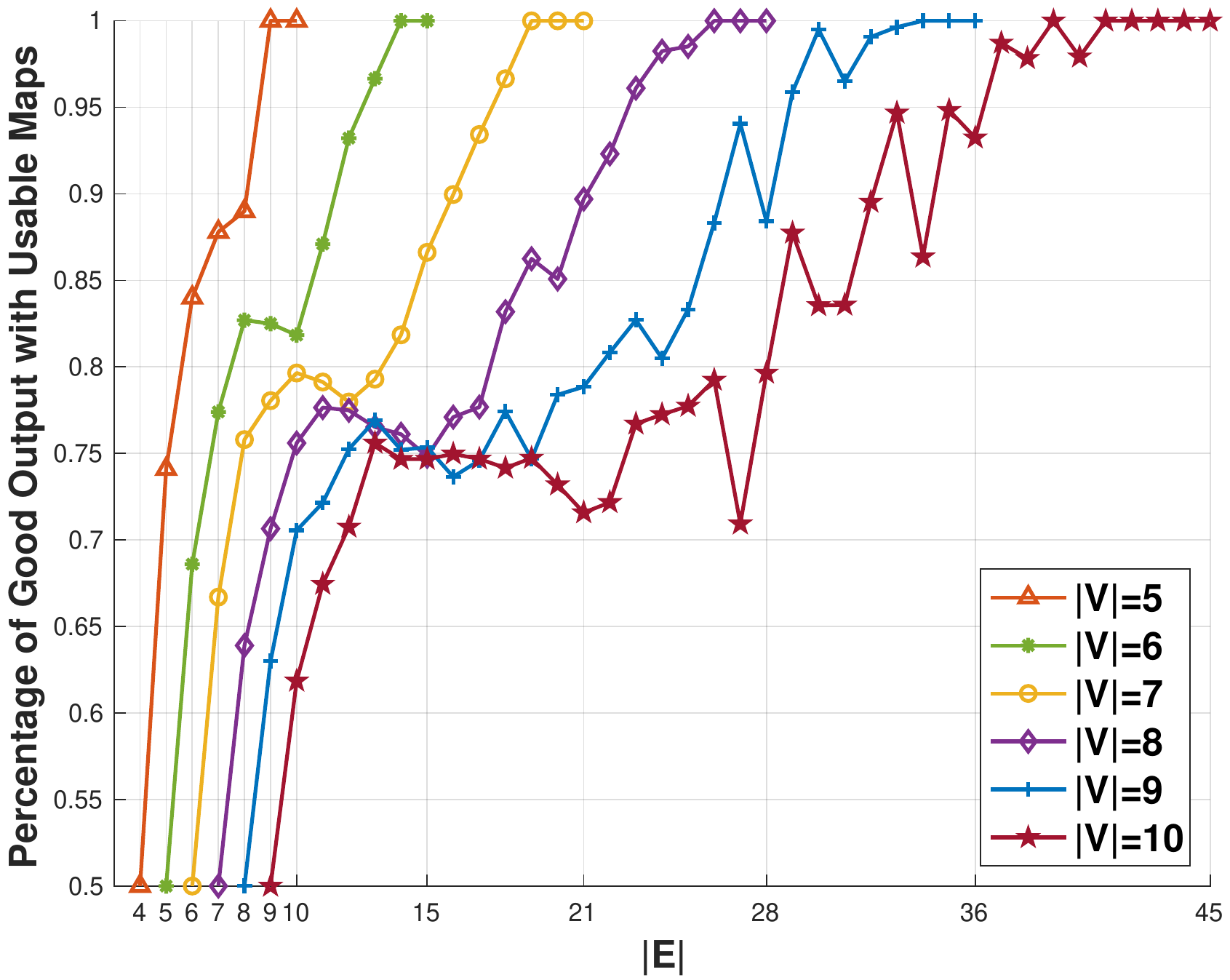}
		\label{subfig:6:good_output_no_dp}
	}
	\hfill
	\subfloat[Percentage of Overall Good Output]{%
		\includegraphics[width=0.32\textwidth]{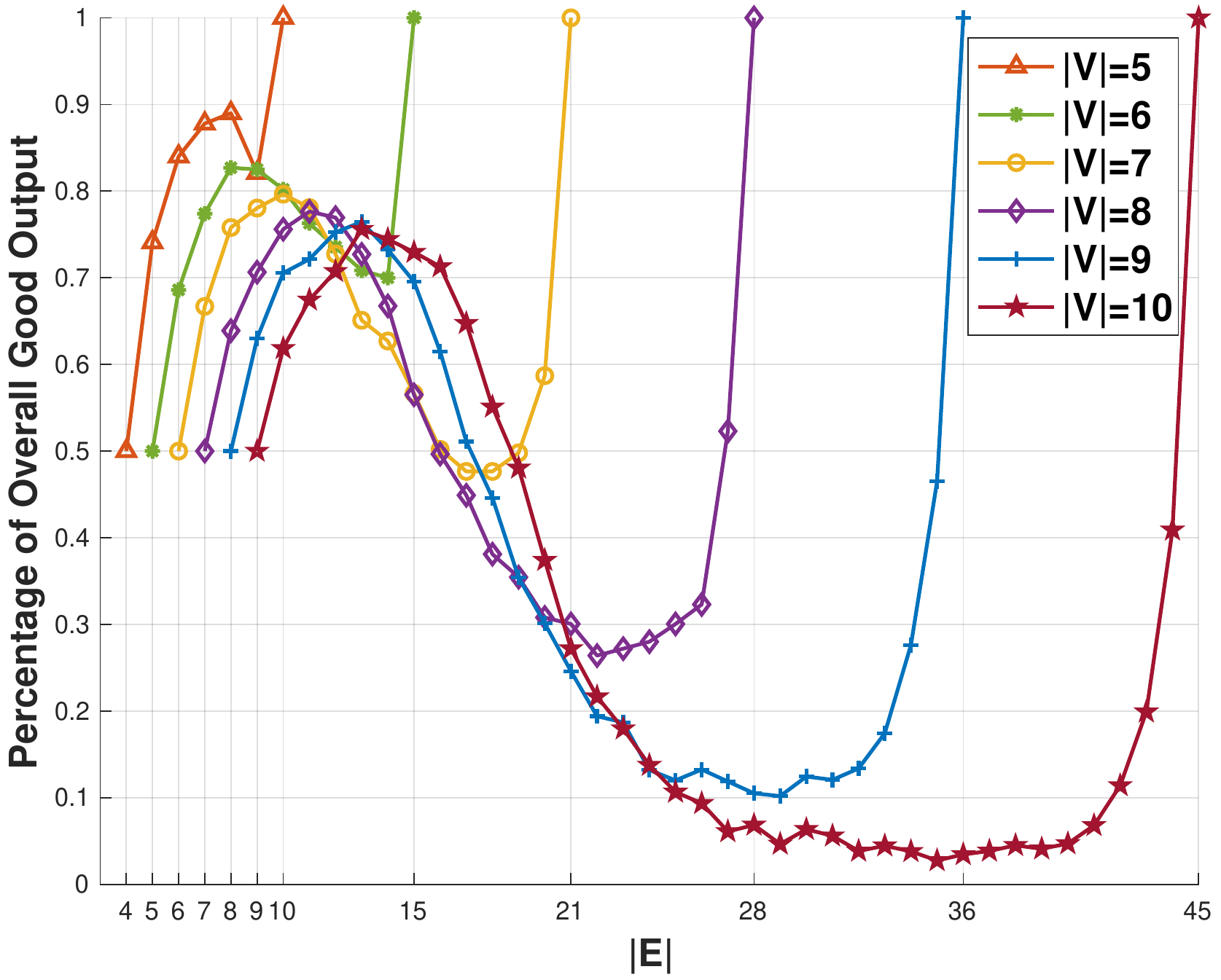}
		\label{subfig:6:overall_good_output_no_dp}
	}
	\caption{Output Quality without both DP and The Splitting Strategy.}
	\label{fig:6:output_no_dp}
\end{figure*}
In Figure~\ref{fig:6:output_no_dp}, we study the performance of our graph-based path without both DP injection and the splitting strategy (as Algorithm~\ref{alg:6:split}) for different sizes of maps with acyclic path. 

Figure~\ref{subfig:6:good_map_no_dp} shows the percentage of usable maps for successfully producing the graph-based path without the splitting strategy. Particularly, about 43\% (52/122) of the maps cannot produce the graph-based path with over 50\% possibilities, especially for most of the maps (paths) with more than 7 vertices. Moreover, only about 30\% (30/122) maps produce the graph-based path with 100\% guarantee, such maps are the complete graph and the sparse graphs $G=(V,E)$, that $|E| \in [|V| - 1, |V| + 3]$.

Figure~\ref{subfig:6:good_output_no_dp} depicts the percentage of good output, which can confirm the correct visiting order of the original path, with the usable maps. According to how we recover the exact path in Section~\ref{sec:6:overview} from the graph-based path, in our experiments we count the good outputs as follows: if we confirm the first node of the original path directly, we count it as a good output; if the two ends of the original path have equal chance to be confirmed as the first node, we add 0.5 to the total number of good outputs. Consequently, we have consistent results from both Figure~\ref{subfig:6:good_output_no_dp} and Proposition~\ref{prop:6:firstvertex}. That is, with at least 50\% probability, our graph-based path confirms the visiting order of the original path, so the exact original path, if we have full information about the vertices and the existence of the edges of a map. In addition, Figure~\ref{subfig:6:good_output_no_dp} also shows that with a given number of vertices, the number of edges and the percentage of good outputs are roughly positively correlated. The reason is that, in average, more edges in the map means more vertices linked to the first node of the path, then recall Line 3 and Line 10 of Algorithm~\ref{alg:6:insert}, we will have higher probability to keep the first node as a root, so to have a good output. The result illustrated in Figure~\ref{subfig:6:good_output_no_dp} will be a base result for analysing the experimental results in the next section.

Finally, Figure~\ref{subfig:6:overall_good_output_no_dp} illustrates the Hadamard product of Figure~\ref{subfig:6:good_map_no_dp} and Figure~\ref{subfig:6:good_output_no_dp}. Namely, the overall good output percentage without the splitting strategy and DP injection.

\subsubsection{Quality of Differentially Private Outputs}
In this part, we report the output quality of our algorithm (with consideration of DP injection and the splitting strategy) on both acyclic paths (in Figure~\ref{fig:6:no_ring_dp}) and cyclic paths (in Figure~\ref{fig:6:ring_dp_v}). With the splitting strategy (Algorithm~\ref{alg:6:split}), we always produce the graph-based path from any map with 100\% guarantee in the experiments, so we do not plot a figure to present it.

\begin{figure*}[!ht]
	\subfloat[DP on Vertices]{%
		\includegraphics[width=0.46\textwidth]{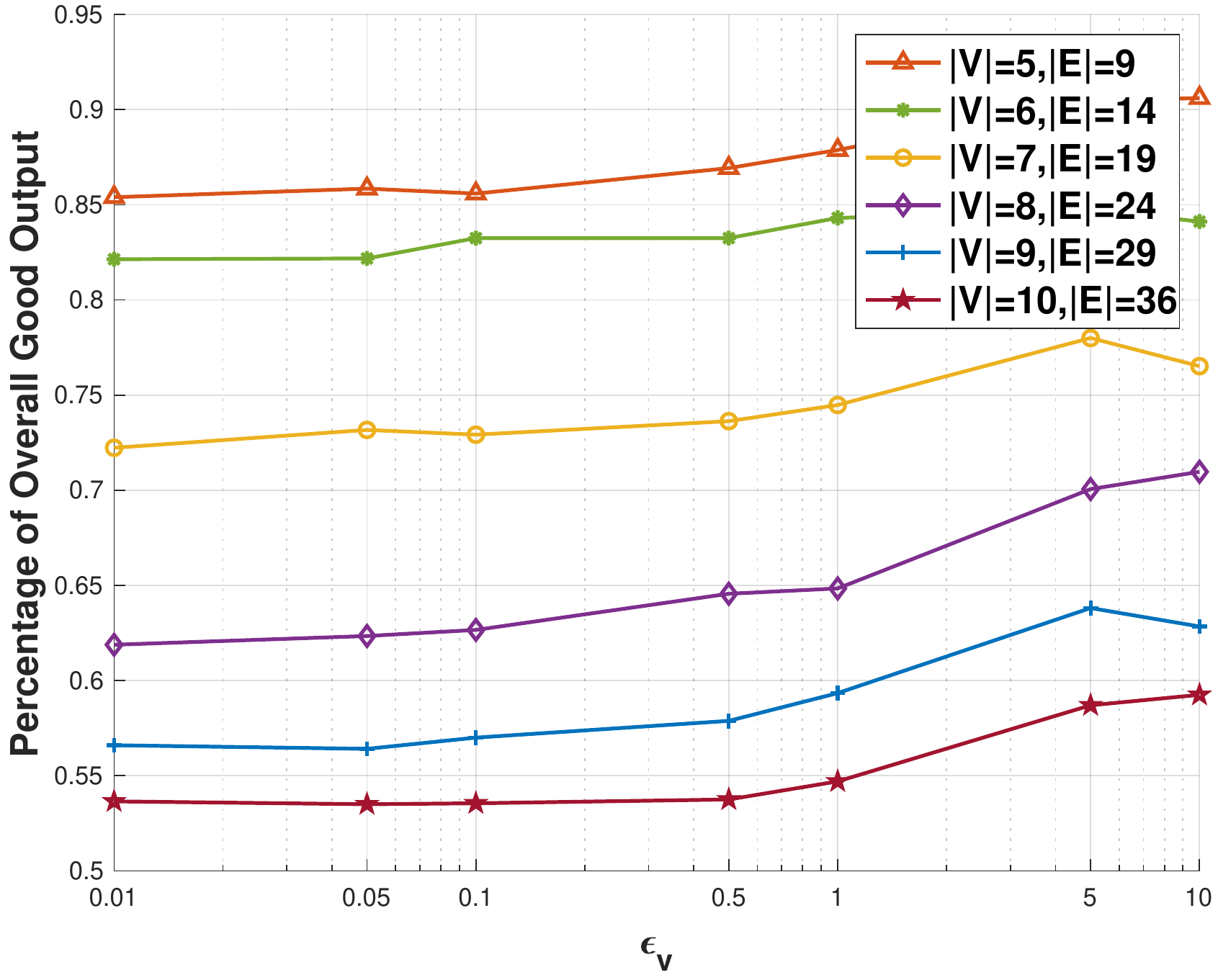}
		\label{subfig:6:no_ring_dp_v}
	}
	\hfill
	\subfloat[DP on Edges]{%
		\includegraphics[width=0.46\textwidth]{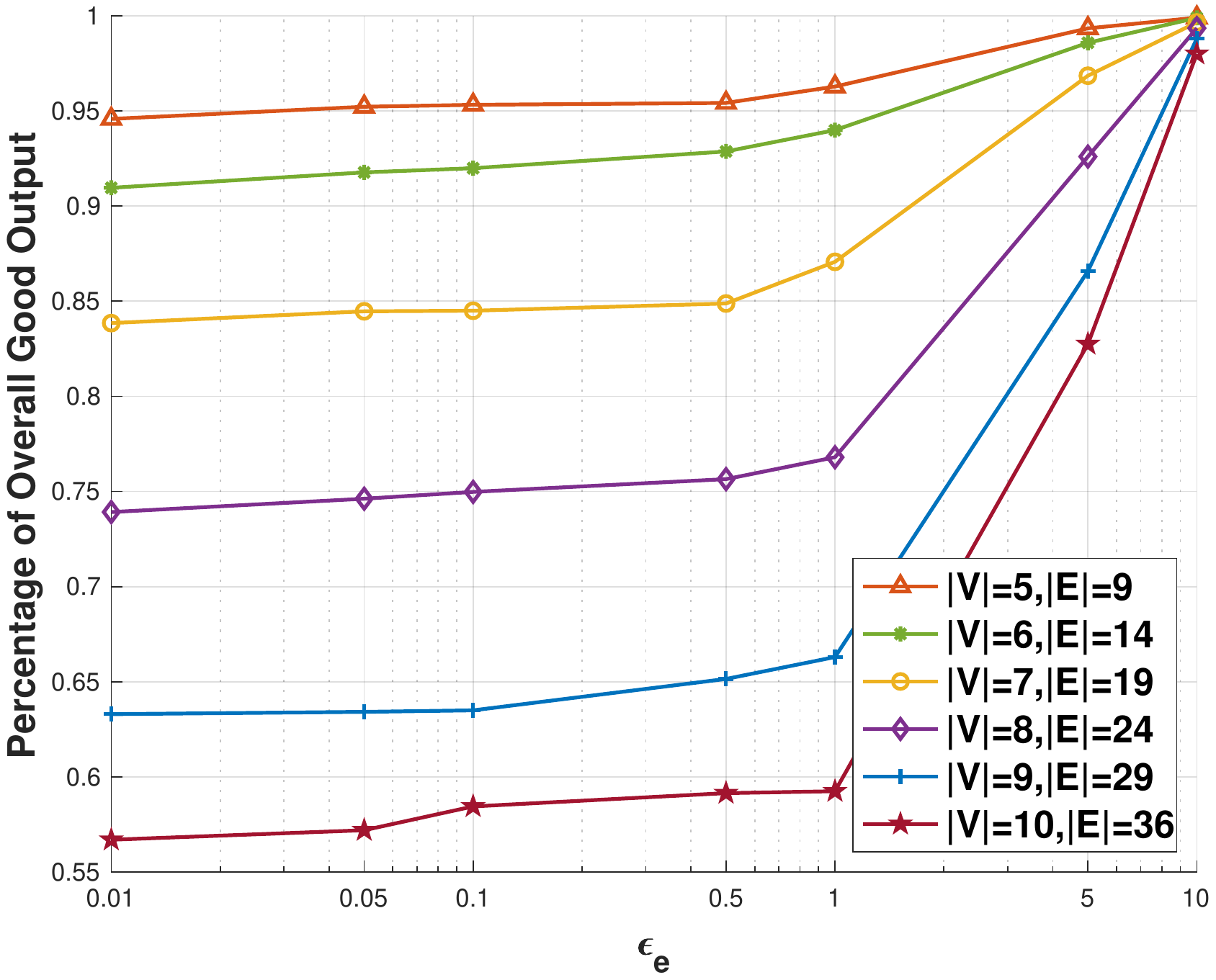}
		\label{subfig:6:no_ring_dp_e}
	}
	\caption{Output Quality with DP on Acyclic Paths.}
	\label{fig:6:no_ring_dp}
\end{figure*}
To evaluate the performance of our algorithm on acyclic paths in the worst case, according to Figure~\ref{subfig:6:good_map_no_dp}, we select the maps produce the graph-based paths with the minimum opportunity for each a given number of vertices in Figure~\ref{fig:6:no_ring_dp}.

In Figure~\ref{subfig:6:no_ring_dp_v}, we only inject DP on vertices to learn the effect of the differentially private vertices pre-processing (Algorithm~\ref{alg:6:dpv}). Based on Equation~\eqref{eq:6:qv} and Figure~\ref{subfig:6:good_output_no_dp}, with an increasing $\epsilon_{v}$, less fake vertices are added, then we will have an increasing percentage of good output, which fits the results in Figure~\ref{subfig:6:no_ring_dp_v}.

Figure~\ref{subfig:6:no_ring_dp_e} shows the performance of differentially private edges pre-process (Algorithm~\ref{alg:6:dpe}). Because a greater $\epsilon_{e}$ introduces more edges to the map ($G$ in Figure~\ref{fig:6:example}), then according to the result of Figure~\ref{subfig:6:good_output_no_dp}, we expect an increasing percentage of good output as what we have in Figure~\ref{subfig:6:no_ring_dp_e}.

Comparing the results between Figure~\ref{subfig:6:no_ring_dp_v} and Figure~\ref{subfig:6:no_ring_dp_e}, because injecting DP on vertices will introduce fake vertices, it decreases the probability of placing the first node of the path as a root, then it impacts the performance of the graph-based path more significantly than DP injection on edges. Therefore, for acyclic path publishing, to protect private edges, only injecting DP on edges would be sufficient for both privacy and utility requirements.

\begin{figure*}[!ht]
	\subfloat[$\epsilon_{e} = 0.5$]{%
		\includegraphics[width=0.32\textwidth]{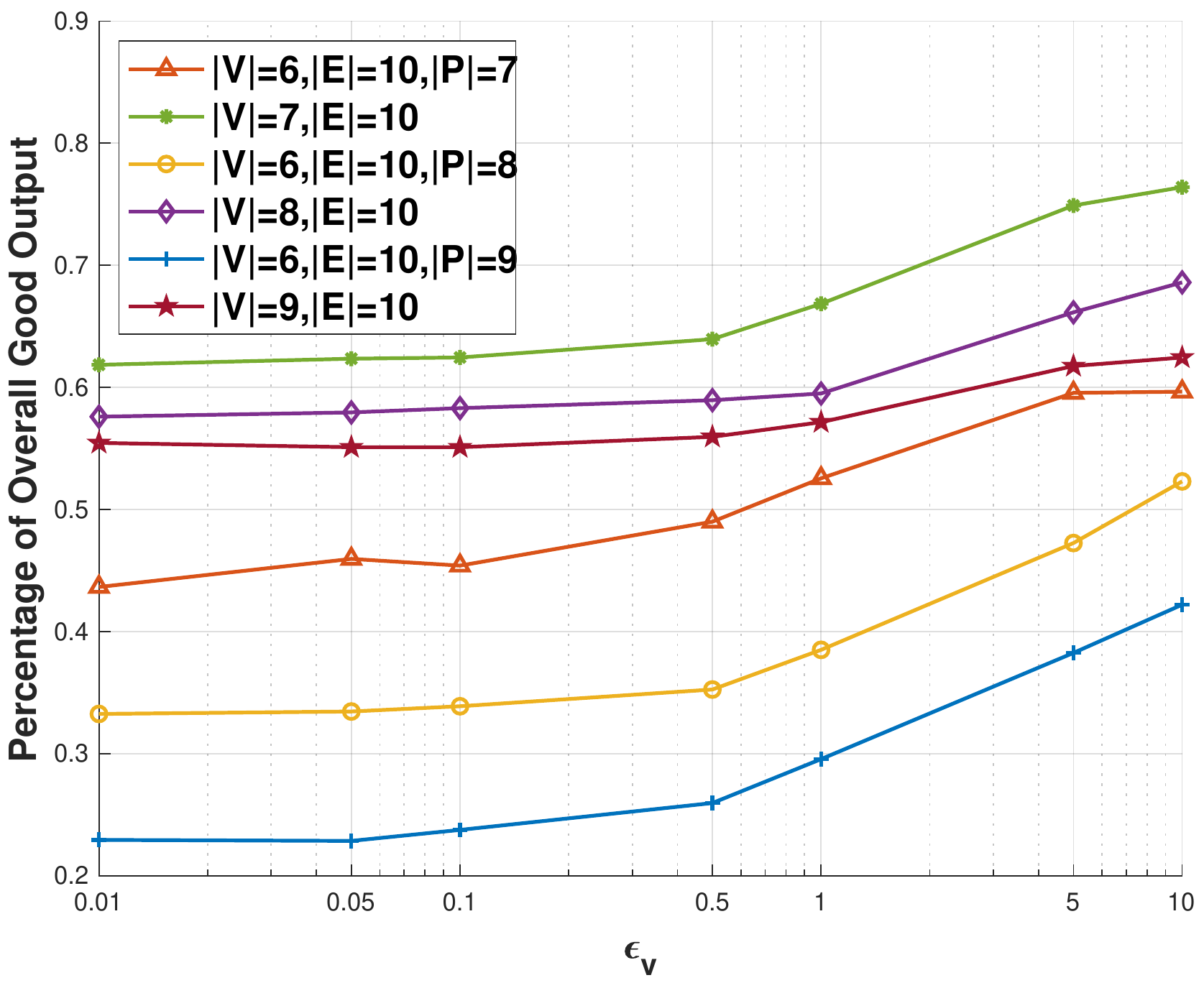}
		\label{subfig:6:ring_dp_ve0}
	}
	\hfill
	\subfloat[$\epsilon_{e} = 1$]{%
		\includegraphics[width=0.32\textwidth]{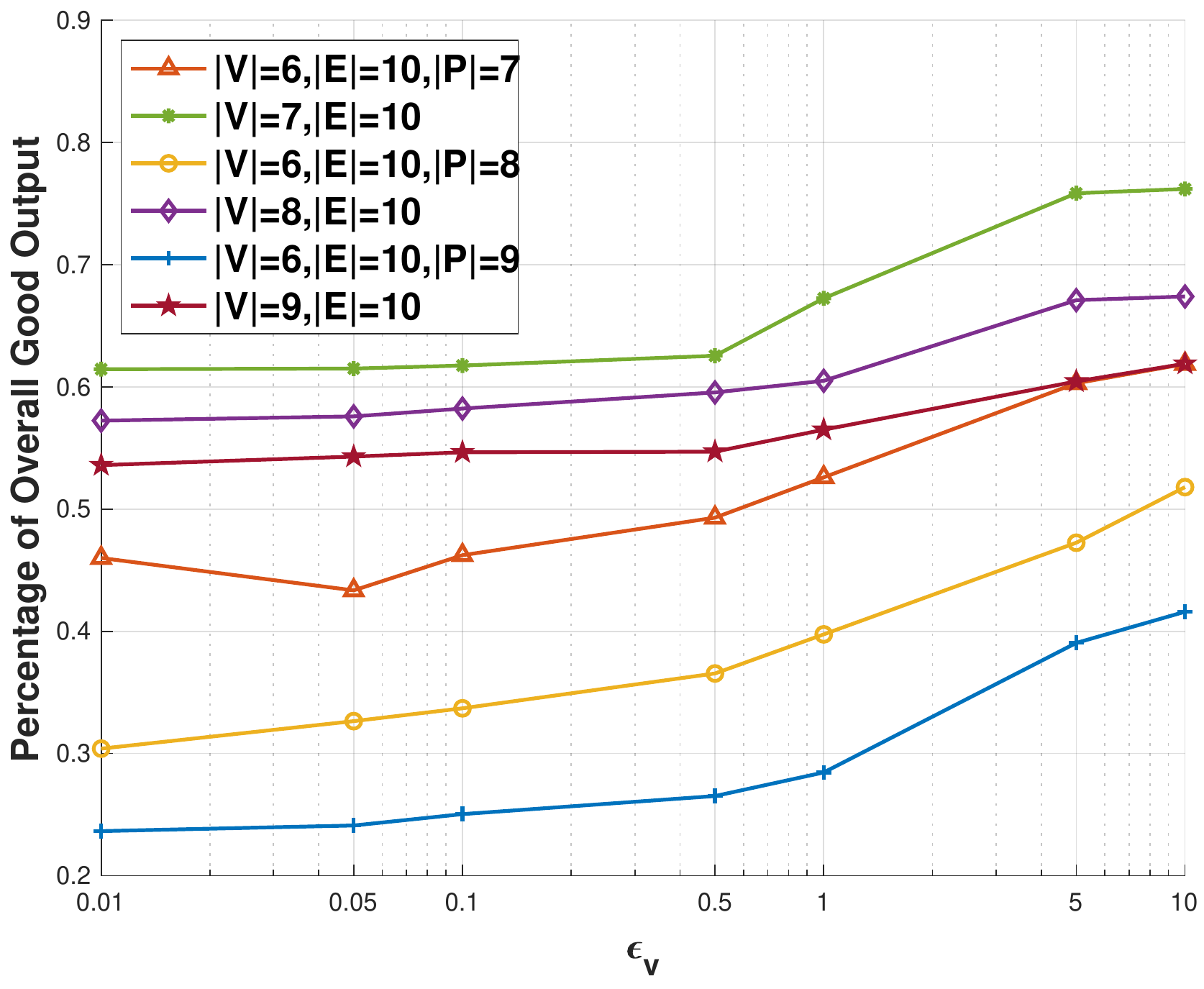}
		\label{subfig:6:ring_dp_ve1}
	}
	\hfill
	\subfloat[$\epsilon_{e} = 5$]{%
		\includegraphics[width=0.32\textwidth]{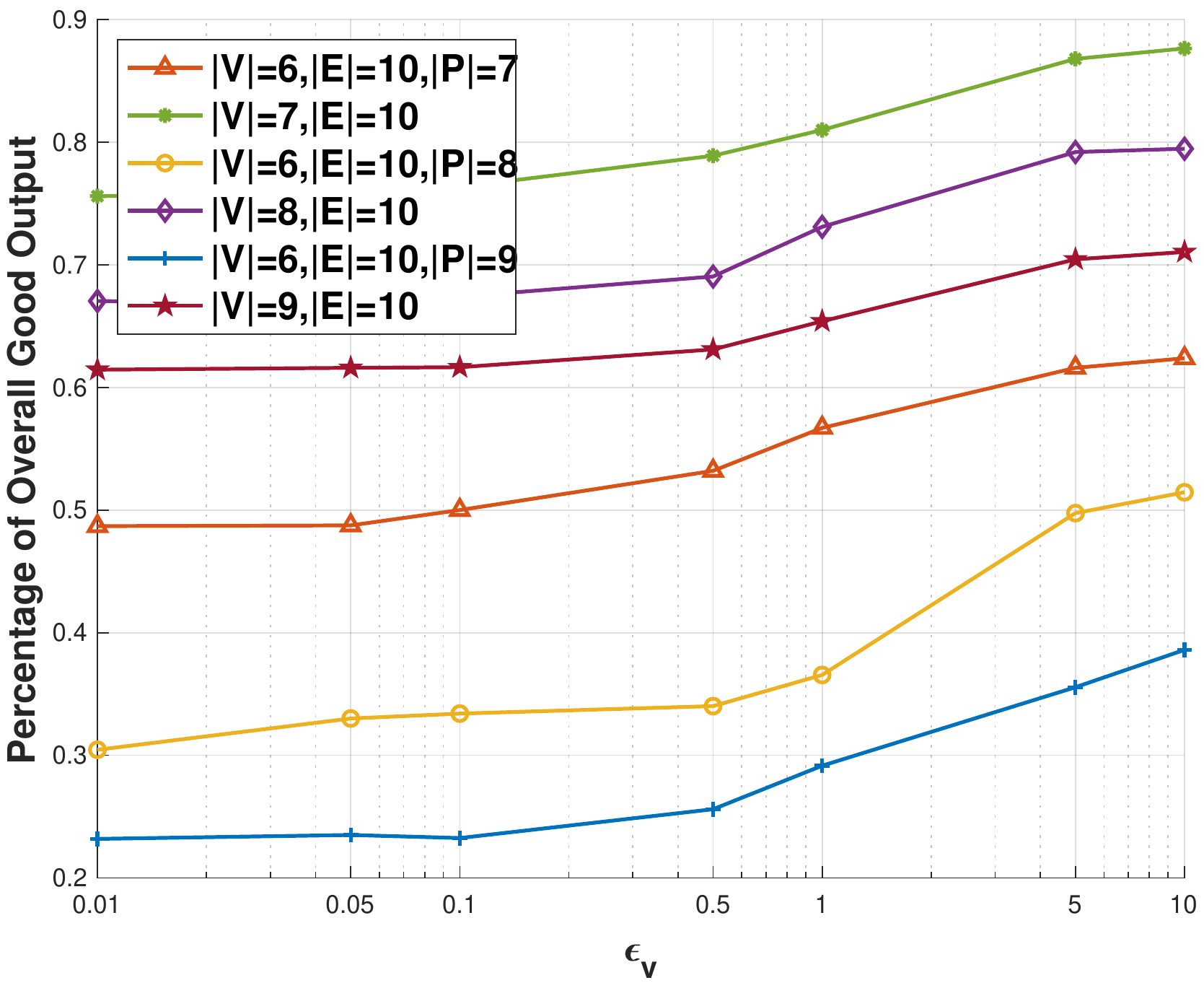}
		\label{subfig:6:ring_dp_ve5}
	}
	\caption{Output Quality with DP on Cyclic Paths.}
	\label{fig:6:ring_dp_v}
\end{figure*}
When we have a cyclic path in a given map, to preserve privacy about the path (the existent of a ring and the existent of an arbitrary edge), DP on vertices and edges are always needed. Since we already studied the performance of DP on vertices in Figure~\ref{subfig:6:no_ring_dp_v}, we inject DP on both vertices and edges in Figure~\ref{fig:6:ring_dp_v}. Specifically, according to Figure~\ref{subfig:6:no_ring_dp_e}, we tuning $\epsilon_{v}$ with three fixed $\epsilon_{e}$s: $\epsilon_{e} = 0.5$ in Figure~\ref{subfig:6:ring_dp_ve0}, $\epsilon_{e} = 1$ in Figure~\ref{subfig:6:ring_dp_ve1}, and $\epsilon_{e} = 5$ in Figure~\ref{subfig:6:ring_dp_ve5}.

In Figure~\ref{fig:6:ring_dp_v}, we select the map with 6 vertices and 10 edges (because it achieves a local minimum for the percentage of good output) as a base map for cyclic path publishing. We compare the performance between a cyclic path and its corresponding acyclic path with same number of vertices. With a give $\epsilon_{v}$ and a given $\epsilon_{e}$, our algorithm achieves better performance on acyclic paths over cyclic paths. Because the vertices in a cyclic path have more than one copies, according to Equation~\eqref{eq:6:qv}, our differentially private vertices pre-processing would create more vertices for cyclic paths than acyclic paths. Then based on the result in Figure~\ref{subfig:6:good_output_no_dp}, more vertices damage the performance of the graph-based path, such results in Figure~\ref{fig:6:ring_dp_v} fit our expectation. Additionally, as the same reason behind Figure~\ref{fig:6:no_ring_dp}, when increasing $\epsilon_{v}$ and $\epsilon_{e}$, the percentage of good output is increased in Figure~\ref{fig:6:ring_dp_v}.

\section{Conclusion}
\label{sec:6:conclusion}
In this paper, we propose a differentially private graph-based path publishing algorithm. In our algorithm, we hide the real path in the topology of the original network by embedding it in a graph containing fake edges and vertices applying the differential privacy technique, so that only the trusted path participants who have the full knowledge about the existence of the edges in a network (which contains the target path) can reconstruct the exact vertices and edges of the original path, and determine the visiting order of the vertices on the path with high probability. Our algorithm effectively protects the edge existence, which is unable to achieve with the existing approaches for differentially private path publishing which preserve the existence of vertices. Both theoretical analysis and experimental evaluation show that our algorithm guarantees both utility and privacy against the adversaries who miss partial edges information but have all the vertices information on the path.


\section*{Acknowledgements}
This work is supported by Australian Government Research Training Program Scholarship, Australian Research Council Discovery Project DP150104871, National Key R \& D Program of China Project \#2017YFB0203201, and supported with supercomputing resources provided by the Phoenix HPC service at the University of Adelaide. The corresponding author is Hong Shen.

\bibliographystyle{plain}
\bibliography{zglu.bib}

\end{document}